\pdfoutput=1
\documentclass[11pt]{article}

\usepackage[utf8]{inputenc}
\usepackage[T1]{fontenc}
\usepackage{lmodern}
\usepackage{microtype}

\usepackage[letterpaper, margin=1in]{geometry}
\usepackage{indentfirst}

\usepackage{amsmath, amssymb, amsthm}
\usepackage{mathrsfs}
\usepackage{mathtools}

\usepackage{algorithm}
\usepackage{algpseudocode}
\algrenewcommand\algorithmicrequire{\textbf{Input:}}
\algrenewcommand\algorithmicensure{\textbf{Output:}}
\usepackage{graphicx}
\usepackage{subcaption}
\usepackage{booktabs}
\usepackage{enumitem}
\usepackage{xcolor}

\usepackage{authblk}

\definecolor{linkblue}{RGB}{0,51,153}
\usepackage[
    colorlinks=true,
    linkcolor=linkblue,
    citecolor=linkblue,
    urlcolor=linkblue,
    bookmarksnumbered=true,
    bookmarksopen=true,
]{hyperref}
\usepackage{cleveref}
\hypersetup{
    pdftitle={Mutual Information Constrained Chernoff Bottleneck},
    pdfauthor={Derek},
    pdfkeywords={information bottleneck, Chernoff information, rate-distortion theory, cardinality bound, Blahut-Arimoto algorithm},
}

\theoremstyle{plain}
\newtheorem{theorem}{Theorem}[section]
\newtheorem{lemma}[theorem]{Lemma}

\newtheorem{corollary}[theorem]{Corollary}
\newtheorem{fact}[theorem]{Fact}

\newtheorem{problem}{Problem}

\theoremstyle{definition}

\newtheorem{example}[theorem]{Example}

\theoremstyle{remark}
\newtheorem{remark}[theorem]{Remark}

\title{\Large\bfseries Mutual Information Constrained Chernoff Bottleneck}

\author{Dier Tang\thanks{Email: \texttt{tangde\_math@connect.hku.hk}.},\;
and Guangyue Han\thanks{Email: \texttt{ghan@hku.hk}.}\\
Department of Mathematics, The University of Hong Kong
}

\date{}

\begin{document}

\maketitle

\begin{abstract}
The classical information bottleneck (IB) measures the relevance of a representation $U$ of $X$ to a target $Y$ by $I(U;Y)$, which does not directly characterize the error of downstream decisions. For a binary hypothesis $Y$ inferred from many separately encoded observations, the optimal error exponent is the Chernoff information between the two conditional distributions of $U$ given $Y$. We study the mutual information constrained Chernoff bottleneck, which seeks an encoder that maximizes this Chernoff information subject to a rate constraint $I(U;X) \leq R$. We show that its optimal value $C(R)$ increases strictly up to $R = H(V)$, where $V$ merges the symbols of $X$ with equal likelihood ratio, remains at the uncompressed exponent beyond, and, unlike the IB curve, need not be concave. We further show that $k+1$ outputs suffice to attain $C(R)$, where $k$ is the cardinality of $V$. We propose an alternating algorithm that updates the encoder via a generalized Blahut--Arimoto algorithm and the Chernoff parameter $s$ via a nonlinear equation, and prove that its iterates remain feasible, with nondecreasing and convergent Chernoff information. Comparing with the classical IB for $n$ samples, we show that, whenever the two objectives disagree, the IB encoder has the smaller equivocation $H(Y|U^{n})$ at $n = 1$, whereas the Chernoff bottleneck has the smaller error probability and equivocation for all large $n$. Numerical experiments confirm the theory, and on real topic-detection data from the 20 Newsgroups corpus, compressing each word to only $17\%$ of its entropy retains $90\%$ of the error exponent and nearly the accuracy of the uncompressed classifier.

\textbf{Keywords:} Information bottleneck; Chernoff information; rate--distortion theory; cardinality bound; hypothesis testing; generalized Blahut--Arimoto algorithm
\end{abstract}

\section{Introduction}\label{sec: introduction}

Representation learning, a fundamental paradigm in machine learning \cite{bengio2013}, aims to automatically learn meaningful and compact representations of raw data that make downstream tasks---such as inference and prediction---more effective \cite{chen2020simple, radford2021learning, shwartzziv2024compress, simeoni2025dinov3, tang2026, luthra2025selfsupervised}. The information bottleneck (IB) \cite{tishby2000} operationalizes this goal by seeking a representation $U$ that preserves the information relevant to a target $Y$ while discarding redundant details of the observation $X$, expressed as the trade-off between maximizing $I(U;Y)$ and minimizing $I(U;X)$. Since then, the IB principle has been used both as a tool for analyzing deep neural networks and as an objective for learning robust, generalizable representations \cite{tishby2015, shwartz2017, alemi2017, kolchinsky2019nonlinear, achille2018, shamir2010, vera2018}.

However, while the mutual information $I(U;Y)$ quantifies average statistical dependence, it does not directly characterize the fundamental limit on the error probability of downstream decision-making. Specifically, when $Y \in \{0,1\}$ is a binary hypothesis to be inferred from $n$ representation samples $U_{1}, \dots, U_{n}$, obtained by applying the same encoder $P_{U|X}$ to observations $X_{1}, \dots, X_{n}$ that are i.i.d.\ given $Y$, the minimum Bayes error probability decays exponentially in $n$, and the optimal error exponent is the Chernoff information \cite{chernoff1952}
\[
C(P_{0}, P_{1}) = \max_{0 \leq s \leq 1} \Big(- \log \sum_{u} P_{0}(u)^{s} P_{1}(u)^{1-s} \Big)
\]
between the conditional distributions $P_{U|0}$ and $P_{U|1}$ of $U$ given $Y = 0$ and $Y = 1$, for every prior with $P_{Y}(0), P_{Y}(1) > 0$ \cite{cover2006, dembo1998}. In contrast to $I(U;Y)$, which depends on the prior, this exponent depends only on $P_{U|0}$ and $P_{U|1}$, and an encoder that is optimal for the classical IB need not be optimal for testing. This operational interpretation motivates replacing $I(U;Y)$ by the Chernoff information, thereby leading to the \emph{mutual information (MI) constrained Chernoff bottleneck}
\[
C(R) = \max_{P_{U|X}} \ C\big(P_{U|0}, P_{U|1}\big) \quad \text{s.t.} \quad I(U;X) \leq R,
\]
which seeks the representation under which the error of the downstream test decays fastest, subject to a budget $R$ on its complexity. The encoder is applied to each sample separately, as when the samples are observed by different sensors or edge devices, or processed by a per-sample feature extractor before the decision is made \cite{tsitsiklis1993dd, chamberland2003, shao2022}, in contrast to testing under communication constraints, where a whole block of observations is encoded jointly \cite{ahlswede1986, han1987}.

Although the MI constrained Chernoff bottleneck has the same form as the classical IB, the change of objective means that the theory and algorithms developed for the IB do not carry over directly. They rely on two properties of $I(U;Y)$. First, the IB loss $I(X;Y) - I(U;Y) = \mathbb{E}\big[\mathrm{D}(P_{Y|X} \| P_{Y|U})\big]$ is an expectation over $(X,U)$ of a per-letter divergence, which leads to self-consistent equations and Blahut--Arimoto-type iterations \cite{tishby2000, blahut1972, arimoto1972}. Second, $I(U;Y) = H(Y) - \sum_{u} P_{U}(u) H(Y|U=u)$ is a sum over the outputs of the encoder, weighted by their probabilities, which underlies the time-sharing and cardinality arguments \cite{witsenhausen1975}. The Chernoff information has neither property: it is not an expectation of a per-letter quantity, and it involves an optimization over the parameter $s$ that couples all outputs. Consequently, neither rate--distortion theory \cite{shannon1948, berger1971} nor the algorithms for the classical IB apply directly, and the shape of the optimal curve, cardinality bounds on the representation, and the computation of optimal encoders must be studied anew.

In this paper, we address these problems. Our contributions are as follows.
\begin{enumerate}[label=(\roman*)]
    \item \textit{Shape of the curve.} We prove that $C(R)$ is strictly increasing on $[0, H(V)]$ and saturates at $C(P_{X|0}, P_{X|1})$ for all $R \geq H(V)$, where $V$ merges the symbols of $X$ with equal likelihood ratio (Theorem~\ref{the: saturation and strict monotonicity}), and that for $R \leq H(V)$ every optimal encoder uses the full rate budget (Corollary~\ref{cor: rate constraint active}). In contrast to the IB curve, $C(R)$ need not be concave (Example~\ref{exa: not concave}).
    
    \item \textit{Cardinality bounds.} We prove that an optimal encoder exists with at most $|\mathcal{X}| + 1$ outputs (Theorem~\ref{the: supp U leq X+1}), and sharpen the bound to $k + 1$, where $k$ is the number of distinct likelihood ratios, by a reduction to $V$ (Theorem~\ref{the: sharpened cardinality}).
    
    \item \textit{Computation.} We develop an algorithm (Algorithm~\ref{alg: Algorithm for reformulated problem}) that alternately updates the encoder, by a generalized Blahut--Arimoto algorithm (Algorithm~\ref{alg: GBA for K_U|X}) with a bisection over the trade-off, and the parameter $s$, by solving a nonlinear equation. We prove that Algorithm~\ref{alg: GBA for K_U|X} converges to a global minimizer at rate $O(1/t)$ with a stopping certificate (Theorem~\ref{the: hatL global converge}, Corollary~\ref{cor: GBA stopping certificate}), that the bisection has an explicit bracket and certificate (Lemma~\ref{lem: bisection search of lambda}), and that the iterates of Algorithm~\ref{alg: Algorithm for reformulated problem} remain feasible, with $s$ uniformly bounded away from $0$ and $1$ and the Chernoff information monotonically convergent (Lemma~\ref{lem: interiority of s^*}, Theorem~\ref{the: convergence of algorithm_reformulated problem}, Remark~\ref{rem: s uniformly bounded away from 0 1}).
    
    \item \textit{Comparison with the IB.} We show that the equivocation $H(Y|U^{n})$ decays at the same exponential rate as the error probability (Lemma~\ref{lem: equivocation vs error}), so that, when the two objectives disagree, the IB encoder, which is the best for a single sample under log-loss, is overtaken by the Chernoff bottleneck as $n$ grows (Theorem~\ref{the: IB vs CB}). For the error probability, the IB need not be better even for small $n$ (Remark~\ref{rem: Pe small n}).
    
    \item \textit{Experiments.} On synthetic sources, the computed $C(R)$ curves confirm the theory of Section~\ref{sec: global shape}, and the optimal encoders are randomized likelihood-ratio quantizers (Sections~\ref{sec: exp C(R)} and~\ref{sec: exp structure}). On real data, in topic detection on the 20 Newsgroups corpus \cite{lang1995}, the Chernoff bottleneck compresses each word to $1$ nat, only $17\%$ of its entropy, while keeping $90\%$ of the error exponent and nearly the accuracy of the uncompressed classifier (Section~\ref{sec: exp real}). In a comparison with the IB, the IB has the smaller error only while the error is about $0.2$ or larger, and below that the Chernoff bottleneck needs $18$--$25\%$ fewer samples (Section~\ref{sec: exp IB}).
\end{enumerate}

The rest of the paper is organized as follows. Section~\ref{sec: Setup} formulates the problem. Section~\ref{sec: global shape} characterizes the shape of the $C(R)$ curve, and Section~\ref{sec: Cardinality Bounds and Attainability} establishes the cardinality bounds and reformulates the problem over a finite output alphabet. Section~\ref{sec: An Iterative Algorithm} develops and analyzes the algorithm, Section~\ref{sec: CB vs IB} compares the Chernoff bottleneck with the IB for a finite number of samples, Section~\ref{sec: numerical experiments} reports the numerical experiments, and Section~\ref{sec: conclusion} concludes the paper.

\section{Problem Setup}\label{sec: Setup}

Let $X$, $Y$ be discrete random variables with finite alphabet $\mathcal{X}$ and $\mathcal{Y} = \{0,1\}$, and joint distribution $P_{X,Y}$ with marginals $P_{X}(x) > 0, P_{Y}(y) > 0$ for any $x \in \mathcal{X}$ and $y \in \{0,1\}$. We derive $P_{X|Y=0}$ and $P_{X|Y=1}$ via Bayes formula, denote as $P_{X|0}$ and $P_{X|1}$ respectively. Assume there are no symbols that are ``instantly recognizable'', i.e.,
\begin{equation}
\label{eq: condition no instantly recognizable x}
    P_{X|0}(x) > 0, \ P_{X|1}(x) > 0, \ \forall x \in \mathcal{X},
\end{equation}
together with the nondegeneracy assumption $P_{X|0} \not\equiv P_{X|1}$.

Suppose samples $X^{n} = (X_1, \cdots, X_{n})$ are i.i.d. drawn from either $P_{X|0}$ or $P_{X|1}$, the error exponent of using $X^{n}$ to infer $Y$ is given by the Chernoff information
\[
C(P_{X|0}, P_{X|1}) = \max_{0 \leq s \leq 1} \Big(- \log \sum_{x} P_{X|0}(x)^{s} P_{X|1}(x)^{1-s} \Big) = - \min_{0 \leq s \leq 1} \log\sum_{x} P_{X|0}(x)^{s} P_{X|1}(x)^{1-s}.
\]
In this paper, we set $0^0 = \lim\limits_{s \to 0} 0^{s} = 0$. All logarithms are natural.

For an encoder $K = K_{U|X}$ that compresses $X \to U$ with finite output alphabet, denote by $P^{K}_{U|0}$, $P^{K}_{U|1}$ the induced conditional distributions, $P^{K}_{U}$ the marginal, and $I_{K}(U;X)$ the mutual information. Suppose samples $U^{n} = (U_1, \cdots, U_{n})$ are i.i.d. drawn from either $P^{K}_{U|0}$ or $P^{K}_{U|1}$, the error exponent of using $U^{n}$ to infer $Y$ is given by the Chernoff information $C(P^{K}_{U|0}, P^{K}_{U|1})$. Define the distortion of the encoder $K_{U|X}$ as
\begin{equation}
\label{eq: definition distortion d(K)}
    d(K) = C(P_{X|0}, P_{X|1}) - C(P^{K}_{U|0}, P^{K}_{U|1}).
\end{equation}
The distortion $d(K)$ measures how much information of inferring $Y$ is lost in the encoding process, and is well-defined according to the following lemma.

\begin{lemma}
\label{lem: d(K) geq 0}
    For any encoder $K_{U|X}$, $d(K) \geq 0$.
\end{lemma}

\begin{proof}
    Since $U \leftrightarrow X \leftrightarrow Y$ forms a Markov chain, the proof is straightforward by the data processing inequality of Chernoff information \cite{vanerven2014}. A direct proof via H\"older's inequality is given in Lemma \ref{lem: holder BsK geq BsX}.
\end{proof}

Given rate $R \geq 0$ and alphabet $\mathcal{U}_{M} = \{1,\cdots,M\}$ with finite $M$, define the \emph{feasible encoder set}
\begin{equation}
\label{eq: feasible encoder set K(M)R}
    \mathcal{K}^{(M)}_{R} = \Big\{K_{U|X}: K(u|x) \geq 0, \sum_{u} K(u|x) = 1, \forall (u,x) \in \mathcal{U}_{M} \times \mathcal{X}; \ I_{K}(U;X) \leq R \Big\},
\end{equation}
with its strictly positive subset
\begin{equation}
\label{eq: positive feasible encoder set K(M)R+}
    \mathcal{K}^{(M)}_{R,+} = \Big\{K_{U|X} \in \mathcal{K}^{(M)}_{R}: K(u|x) > 0, \ \forall (u,x) \in \mathcal{U}_{M} \times \mathcal{X} \Big\},
\end{equation}
and define the \emph{entire feasible set}
\begin{equation}
\label{eq: union feasible encoder set}
    \mathscr{K}_{R} = \bigcup_{M \geq 1} \mathcal{K}^{(M)}_{R}.
\end{equation}
Then, our problem can be formulated as
\begin{problem}
\label{prob: initial problem}
Given $R \geq 0$, seek
\[
K^{*} \in \arg \inf_{K \in \mathscr{K}_{R}} d(K),
\]
equivalently,
\begin{equation}
\label{eq: problem 1}
    K^{*} \in \arg \sup_{K \in \mathscr{K}_{R}} C(P^{K}_{U|0}, P^{K}_{U|1}).
\end{equation}
\end{problem}
We refer to \eqref{eq: problem 1} as the \emph{mutual information (MI) constrained Chernoff bottleneck}, for which the attainability of the optimal solutions is established in Section \ref{sec: Cardinality Bounds and Attainability}.

For rate $R \geq 0$, define the supremum in \eqref{eq: problem 1} as
\begin{equation}
\label{eq: definition CR}
    C(R) \stackrel{\Delta}{=} \sup_{K \in \mathscr{K}_{R}} C(P^{K}_{U|0}, P^{K}_{U|1}),
\end{equation}
which is our object of study. The prior $P_{Y}$ enters Problem~\ref{prob: initial problem} only through $P_{X} = P_{Y}(0) P_{X|0} + P_{Y}(1) P_{X|1}$, the distribution of the observations seen by the encoder, which defines the rate $I_{K}(U;X)$; the objective $C(P^{K}_{U|0}, P^{K}_{U|1})$ depends only on $P_{X|0}$ and $P_{X|1}$. Let $F(R) \stackrel{\Delta}{=} e^{-C(R)}$.

Given alphabet $\mathcal{U}_{M} = \{1,\cdots,M\}$ with finite $M$. For $K \in \mathcal{K}^{(M)}_{R}$ and $s \in [0,1]$, define
\begin{equation}
\label{eq: definition BM(K,s)}
    B_{M}(K, s) \stackrel{\Delta}{=} \sum_{u \in \mathcal{U}_{M}} P^{K}_{U|0}(u)^{s} P^{K}_{U|1}(u)^{1-s},
\end{equation}
and for $K \in \mathcal{K}^{(M)}_{R}$, define the corresponding Chernoff information
\begin{equation}
\label{eq: definition Chernoff CK}
    \mathcal{C}(K) \stackrel{\Delta}{=} -\log \min_{s \in [0,1]} B_{M}(K, s).
\end{equation}

We abbreviate $C_{X} = C(P_{X|0}, P_{X|1}) > 0$, which is strictly positive since $P_{X|0} \not\equiv P_{X|1}$. Denote by
\begin{equation}
\label{eq: definition of ell and L}
    \ell \stackrel{\Delta}{=} \min_{x \in \mathcal{X}} \frac{P_{X|1}(x)}{P_{X|0}(x)}, \quad L \stackrel{\Delta}{=} \max_{x \in \mathcal{X}} \frac{P_{X|1}(x)}{P_{X|0}(x)},
\end{equation}
which are both well defined according to \eqref{eq: condition no instantly recognizable x}, and let
\begin{equation}
\label{eq: definition of Gamma}
    \Gamma \stackrel{\Delta}{=} \max \{L, \ell^{-1}\} = \exp \Big(\max_{x} \Big|\log \frac{P_{X|1}(x)}{P_{X|0}(x)} \Big| \Big) > 1.
\end{equation}

For all $x \in \mathcal{X}$, let $\rho(x) \stackrel{\Delta}{=} \dfrac{P_{X|1}(x)}{P_{X|0}(x)} \in [\ell, L]$, and $\{\rho_1, \cdots, \rho_{k} \}$ be the distinct values taken by $\rho$, then $2 \leq k \leq |\mathcal{X}|$ (the lower bound is because $P_{X|0} \not\equiv P_{X|1}$). Define $\nu : \mathcal{X} \to \{1, \cdots, k \}$ by $\nu(x) = j$ if and only if $\rho(x) = \rho_{j}$, and set variable
\begin{equation}
\label{eq: definition V = nu(X)}
    V = \nu(X).
\end{equation}
Thus $V$ merges exactly those symbols that have the same likelihood ratio, equivalently the same posterior $P_{Y|X=x}$; it is the minimal sufficient statistic of $X$ for $Y$ \cite{cover2006}, so that $I(V;Y) = I(X;Y)$. Moreover, $P_{V}(j) = \sum_{x: \nu(x)=j} P_{X}(x)$, and since $V$ is a function of $X$, we have $I(V;X) = H(V) \leq H(X)$, with equality if and only if all $|\mathcal{X}|$ ratios are distinct. We write $\nu^{-1}$ as the inverse mapping of $\nu$, and define
\begin{equation}
\label{eq: definition likelihood ratio classes}
    \mathcal{X}_{j} \stackrel{\Delta}{=} \nu^{-1}(j) = \{x \in \mathcal{X}: \rho(x) = \rho_{j}\}, \quad j = 1, \cdots, k,
\end{equation}
for the \textit{likelihood-ratio classes}, so that $\mathcal{X} = \bigsqcup_{j=1}^{k} \mathcal{X}_{j}$. Denote by $K^{X}$ the identity encoder that outputs $X$ itself (with no compression), and
\begin{equation}
\label{eq: definition KV}
    K^{V}(u|x) = \mathbf{1}\{u = \nu(x)\}, \quad (u,x) \in \mathcal{U}_{k} \times \mathcal{X},
\end{equation}
the deterministic encoder that outputs $V$.

\section{Global Shape of the \texorpdfstring{$C(R)$}{C(R)} Curve}\label{sec: global shape}

With $C(R)$ defined in \eqref{eq: definition CR}, the following lemma shows some basic shape of the $C(R)$ curve.

\begin{lemma}
\label{lem: basic shape of C(R)}
\begin{enumerate}[label=(\roman*)]
    \item\label{item: basic shape of C(R) 1} $C(0) = 0$;
    \item\label{item: basic shape of C(R) 2} $C(R)$ is nondecreasing on $[0, \infty)$;
    \item\label{item: basic shape of C(R) 3} $0 \leq C(R) \leq C_{X}$ for all $R \in [0, \infty)$.
\end{enumerate}
\end{lemma}

\begin{proof}
    For \ref{item: basic shape of C(R) 1}, Markovity and $I_{K}(U;X) = 0$ implies $P^{K}_{U|0} \equiv P^{K}_{U|1}$. Hence, $C(0) = 0$. For \ref{item: basic shape of C(R) 2}, monotonicity holds because $R \mapsto \mathscr{K}_{R}$ is nondecreasing for inclusion. For \ref{item: basic shape of C(R) 3}, the upper bound follows from Lemma \ref{lem: d(K) geq 0}, i.e., the data processing inequality for Chernoff information \cite{vanerven2014, csiszar1967}.
\end{proof}

Theorem \ref{the: saturation and strict monotonicity} gives a stronger result for the shape of $C(R)$, and its proof relies on the following lemmas. Throughout, $B_{M}(K,s)$ and $V$ are defined in \eqref{eq: definition BM(K,s)} and \eqref{eq: definition V = nu(X)}, respectively.

\begin{lemma}[Sufficiency of $V$]
\label{lem: KV preserves Bs}
    Let $|\mathcal{V}| = k$, and $K^{V}$ be defined in \eqref{eq: definition KV}, we have
    \[
    I_{K^{V}}(U;X) = H(V), \quad B_{k}(K^{V},s) = B_{|\mathcal{X}|}(K^{X},s), \ \forall s \in [0,1].
    \]
    In particular, $C(P^{K^{V}}_{U|0}, P^{K^{V}}_{U|1}) = C_{X}$.
\end{lemma}

\begin{proof}
    Appendix~\ref{proof: KV preserves Bs}.
\end{proof}

\begin{lemma}[Hölder]
\label{lem: holder BsK geq BsX}
    For any encoder $K \in \mathcal{K}^{(M)}_{R}$ and $s \in [0,1]$,
    \begin{equation}
    \label{eq: BsK geq BsX}
        B_{M}(K,s) \geq B_{|\mathcal{X}|}(K^{X},s).
    \end{equation}
    Moreover, if $I_{K}(U;X) < H(V)$, then the inequality \eqref{eq: BsK geq BsX} is strict for every $s \in (0,1)$.
\end{lemma}

\begin{proof}
    Appendix~\ref{proof: holder BsK geq BsX}.
\end{proof}

\begin{lemma}[Time-sharing]
\label{lem: time sharing}
    Let $K_1, K_2$ be encoders with output alphabets $\mathcal{U}_{M_1}, \mathcal{U}_{M_2}$, producing $U_1$ and $U_2$ from $X$, respectively. For $\theta \in [0,1]$, let $Q \sim \mathrm{Bern}(\theta)$ be independent of $X$, and let $K$ be the encoder with output alphabet $\mathcal{U}_{M_1+M_2}$ producing $U = U_1$ if $Q = 1$ and $U = M_1 + U_2$ if $Q = 0$. Then, for any $s \in [0,1]$,
    \[
    I_{K}(U;X) = \theta I_{K_1}(U;X) + (1-\theta) I_{K_2}(U;X), \quad B_{M_1+M_2}(K,s) = \theta B_{M_1}(K_1,s) + (1-\theta) B_{M_2}(K_2,s).
    \]
\end{lemma}

\begin{proof}
    Appendix~\ref{proof: time sharing}.
\end{proof}

\begin{theorem}[Saturation and strict monotonicity]
\label{the: saturation and strict monotonicity}
    $C(R)$ is strictly increasing on $[0, H(V)]$, and $C(R) = C_{X}$ for every $R \geq H(V)$.
\end{theorem}

\begin{proof}
    \textit{Saturation.} Let $R \geq H(V)$. By Lemma \ref{lem: KV preserves Bs}, $I_{K^{V}}(U;X) = H(V) \leq R$, so that $K^{V} \in \mathscr{K}_{R}$, and $C(P^{K^{V}}_{U|0}, P^{K^{V}}_{U|1}) = C_{X}$. Hence $C(R) \geq C_{X}$, while Lemma \ref{lem: basic shape of C(R)} gives $C(R) \leq C_{X}$. Therefore $C(R) = C_{X}$.

    \textit{Strict monotonicity.} Let $0 \leq R_1 < R_2 \leq H(V)$, and let $K^{*} \in \mathscr{K}_{R_1}$ with output alphabet $\mathcal{U}_{M^{*}}$ attain $C(R_1)$, i.e., $\mathcal{C}(K^{*}) = C(R_1)$ (the existence of $K^{*}$ is guaranteed by Theorem \ref{the: supp U leq X+1}). Write $I^{*} = I_{K^{*}}(U;X) \leq R_1 < H(V)$.

    By Young's inequality (Fact \ref{fact: Young's inequality}), for every $s \in [0,1]$,
    \[
    B_{M^{*}}(K^{*},s) \leq \sum_{u \in \mathcal{U}_{M^{*}}} \big( s P^{K^{*}}_{U|0}(u) + (1-s) P^{K^{*}}_{U|1}(u) \big) = 1,
    \]
    while $B_{M^{*}}(K^{*},0) = B_{M^{*}}(K^{*},1) = 1$. If $B_{M^{*}}(K^{*},s) \equiv 1$ on $[0,1]$, we select $\bar{s} = \frac{1}{2}$. Otherwise, $\min_{s \in [0,1]} B_{M^{*}}(K^{*},s) < 1$, then every $\bar{s} \in \arg\min_{s \in [0,1]} B_{M^{*}}(K^{*},s)$ lies in $(0,1)$. In both cases, we can select $\bar{s} \in (0,1)$ and
    \begin{equation}
    \label{eq: choice of bar s}
        B_{M^{*}}(K^{*},\bar{s}) = \min_{s \in [0,1]} B_{M^{*}}(K^{*},s) = e^{-\mathcal{C}(K^{*})} = F(R_1).
    \end{equation}
    Moreover, since $I^{*} < H(V)$, Lemma \ref{lem: holder BsK geq BsX} gives
    \begin{equation}
    \label{eq: strict gap at bar s}
        B_{M^{*}}(K^{*},\bar{s}) > B_{|\mathcal{X}|}(K^{X},\bar{s}).
    \end{equation}

    Now let $K'$ be the encoder with output alphabet $\mathcal{U}_{k+M^{*}}$ given by Lemma \ref{lem: time sharing}, with $K_1 = K^{V}$, $K_2 = K^{*}$ and
    \[
    \theta = \frac{R_2 - I^{*}}{H(V) - I^{*}} \in (0, 1].
    \]
    By Lemma \ref{lem: time sharing} and Lemma \ref{lem: KV preserves Bs},
    \[
    I_{K'}(U;X) = \theta H(V) + (1-\theta) I^{*} = R_2,
    \]
    so that $K' \in \mathscr{K}_{R_2}$. Using Lemma \ref{lem: time sharing}, then Lemma \ref{lem: KV preserves Bs}, then \eqref{eq: strict gap at bar s} with $\theta > 0$, and finally \eqref{eq: choice of bar s},
    \[
    \begin{aligned}
    F(R_2) & \leq e^{-\mathcal{C}(K')} \leq B_{k+M^{*}}(K',\bar{s}) = \theta B_{k}(K^{V},\bar{s}) + (1-\theta) B_{M^{*}}(K^{*},\bar{s}) \\
    & = \theta B_{|\mathcal{X}|}(K^{X},\bar{s}) + (1-\theta) B_{M^{*}}(K^{*},\bar{s}) < B_{M^{*}}(K^{*},\bar{s}) = F(R_1).
    \end{aligned}
    \]
    Therefore, $C(R_2) > C(R_1)$, which completes the proof.
\end{proof}

\begin{corollary}[The rate constraint is active]
\label{cor: rate constraint active}
    For $R \leq H(V)$, every optimal encoder of Problem \ref{prob: initial problem} satisfies $I_{K}(U;X) = R$.
\end{corollary}

\begin{proof}
    Assume $K$ is optimal at rate $R$ with $I_{K}(U;X) < R$. Since $K \in \mathscr{K}_{I_{K}(U;X)}$, we have $C(I_{K}(U;X)) \geq C(R)$, contradicting the strict monotonicity of $C$ on $[0, H(V)]$ established in Theorem \ref{the: saturation and strict monotonicity}.
\end{proof}

The IB curve $R \mapsto \max\{I(U;Y) : I(U;X) \leq R\}$ also reaches its maximum $I(X;Y)$ exactly at $R = H(V)$: in both problems, the maximum is attained only when $U$ is sufficient for $Y$, which requires $I(U;X) \geq H(V)$. However, time-sharing guarantees the concavity of the IB curve, whereas $C(R)$ need not be concave, as the following example shows.

\begin{example}[$C(R)$ curve not necessarily concave]
\label{exa: not concave}
    Let $\mathcal{X} = \{1,2\}$, $P_{Y}(0) = P_{Y}(1) = \frac{1}{2}$, $P_{X|0} = (1-a, a)$ and $P_{X|1} = (a, 1-a)$ with $0 < a \leq 10^{-3}$. Then, $C(R)$ is not concave on $[0, \log 2]$.
\end{example}

\begin{proof}
    Appendix~\ref{proof: counterexample}.
\end{proof}

\section{Cardinality Bounds and Attainability}\label{sec: Cardinality Bounds and Attainability}

Since Problem \ref{prob: initial problem} is searching $K$ in the entire feasible set, which is the union of an infinite number of sets, the attainability of the optimal solution has not been determined. Theorem \ref{the: supp U leq X+1} provides a guarantee in this regard, while Lemma \ref{lem: K_R non empty compact convex} shows many good properties of the feasible encoder set $\mathcal{K}^{(M)}_{R}$, and is used in the proof of Theorem \ref{the: supp U leq X+1}.

\begin{lemma}
\label{lem: K_R non empty compact convex}
    Given finite $M$, for any $R \geq 0$, the feasible encoder set $\mathcal{K}^{(M)}_{R}$ is non-empty, compact and convex.
\end{lemma}

\begin{proof}
    Appendix~\ref{proof: K_R non empty compact convex}.
\end{proof}

\begin{theorem}[Cardinality bound]
\label{the: supp U leq X+1}
For any $R \geq 0$,
\begin{equation}
\label{eq: sup = max X+1}
    \sup_{K \in \mathscr{K}_{R}} C(P^{K}_{U|0}, P^{K}_{U|1}) = \max_{K \in \mathcal{K}^{(|\mathcal{X}|+1)}_{R}} C(P^{K}_{U|0}, P^{K}_{U|1}),
\end{equation}
which indicates the cardinality of $U$ can be restricted to satisfy $|\mathcal{U}| \leq |\mathcal{X}| + 1$ without loss of generality. Moreover, the maximum in \eqref{eq: sup = max X+1} is achievable.
\end{theorem}

\begin{proof}
Assume $K \in \mathscr{K}_{R}$ is an arbitrary feasible encoder with output alphabet $\mathcal{U}_{m}$. Let $B_{m}(K, s)$ be defined in \eqref{eq: definition BM(K,s)}, and choose
\begin{equation}
\label{eq: select parameter s_K}
    s_{K} \in \arg \min_{0 \leq s \leq 1} B_{m}(K, s).
\end{equation}
Then, the corresponding Chernoff information is
\[
\mathcal{C}(K) = - \log B_{m}(K, s_{K}).
\]

Assume $P_{U}^{K}(u) > 0$ for all $u \in \mathcal{U}_{m}$, otherwise, it could be directly removed and reduce $m$. For each $u \in \mathcal{U}_{m}$, define $\lambda_{u} \stackrel{\Delta}{=} P_{U}^{K}(u)$, $w_{u} \stackrel{\Delta}{=} P_{X|U=u}^{K} \in \Delta_{|\mathcal{X}|-1}$, and for $s \in [0,1]$,
\[
g_{s}(w) \stackrel{\Delta}{=} \Big( \frac{\sum_{x} w(x) P_{Y|X}(0|x)}{P_{Y}(0)} \Big)^{s} \Big(\frac{\sum_{x} w(x) P_{Y|X}(1|x)}{P_{Y}(1)} \Big)^{1-s}.
\]
Then,
\[
B_{m}(K, s) = \sum_{u=1}^{m} \lambda_{u} g_{s}(w_{u}).
\]
Construct vectors
\[
v_{u} = \big(1, w_{u}(1), \cdots, w_{u}(|\mathcal{X}|-1), \mathrm{D}(w_{u} \| P_{X}) \big) \in \mathbb{R}^{|\mathcal{X}|+1}, \quad \forall u \in \mathcal{U}_{m}.
\]

If $m \geq |\mathcal{X}| + 2$, then vectors $v_1, \cdots, v_{m}$ are linearly dependent, which means there exists $c_1, \cdots, c_{m}$ not all zeros such that $\sum_{u=1}^{m} c_{u} v_{u} = 0$. In this case, expand each coordinate gives
\[
\sum_{u=1}^{m} c_{u} = 0;
\]
\begin{equation}
\label{eq: sum cuwu(x) = 0}
    \sum_{u=1}^{m} c_{u} w_{u}(x) = 0, \quad x=1,\cdots,|\mathcal{X}|-1;
\end{equation}
\[
\sum_{u=1}^{m} c_{u} \mathrm{D}(w_{u} \| P_{X}) = 0.
\]
Moreover, since
\[
\sum_{u=1}^{m} c_{u} w_{u}(x_{|\mathcal{X}|}) = \sum_{u=1}^{m} c_{u} \big(1- \sum_{x=1}^{|\mathcal{X}|-1} w_{u}(x) \big) = \sum_{u=1}^{m} c_{u} - \sum_{x=1}^{|\mathcal{X}|-1} \sum_{u=1}^{m} c_{u} w_{u}(x) = 0,
\]
the equation \eqref{eq: sum cuwu(x) = 0} holds for all $x \in \mathcal{X}$.

For any $t \in \mathbb{R}$, define the disturbance weight
\[
\lambda_{u}(t) = \lambda_{u} + t c_{u}.
\]
Because $c_1, \cdots, c_{m}$ are not all zeros, and $\sum_{u=1}^{m} c_{u} = 0$, there must be both positive and negative numbers in the sequence $\{c_{u}\}_{u=1}^{m}$. Define
\[
t_{-} = \max_{u: c_{u} > 0} \Big( -\frac{\lambda_{u}}{c_{u}} \Big) < 0, \quad t_{+} = \min_{u: c_{u} < 0} \Big( -\frac{\lambda_{u}}{c_{u}} \Big) > 0.
\]
Then, for any $t \in (t_{-}, t_{+})$, all $\lambda_{u}(t)$ remains positive; when $t = t_{-}$ or $t_{+}$, at least one weight $\lambda_{u}(t)$ goes exactly to zero; for any $t \in [t_{-}, t_{+}]$, we have
\[
\sum_{u=1}^{m} \lambda_{u}(t) = \sum_{u=1}^{m} \lambda_{u} + t \sum_{u=1}^{m} c_{u} = 1;
\]
\[
\sum_{u=1}^{m} \lambda_{u}(t) w_{u}(x) = P_{X}(x) + t \sum_{u=1}^{m} c_{u} w_{u}(x) = P_{X}(x);
\]
\[
I_{t}(U;X) = \sum_{u=1}^{m} \lambda_{u}(t) D(w_{u} \| P_{X}) = I_{K}(U;X) + t \sum_{u=1}^{m} c_{u} D(w_{u} \| P_{X}) = I_{K}(U;X).
\]

Furthermore, for the $s_{K}$ selected in \eqref{eq: select parameter s_K}, define
\[
\hat{B}_{s_{K}}(t) \stackrel{\Delta}{=} \sum_{u=1}^{m} \lambda_{u}(t) g_{s_{K}}(w_{u}) = \hat{B}_{s_{K}}(0) + t \sum_{u=1}^{m} c_{u} g_{s_{K}}(w_{u}).
\]
If $\sum_{u=1}^{m} c_{u} g_{s_{K}}(w_{u}) \geq 0$, let $t = t_{-} < 0$; if $\sum_{u=1}^{m} c_{u} g_{s_{K}}(w_{u}) < 0$, let $t = t_{+} > 0$. Thus, we can always select $t_0 \in \{t_{-}, t_{+}\}$ such that
\[
\hat{B}_{s_{K}}(t_0) \leq \hat{B}_{s_{K}}(0) = B_{m}(K, s_{K}),
\]
and at least one $\lambda_{u}(t_0) = 0$. Define the updated weights $\tilde{\lambda}_{u} = \lambda_{u}(t_0), \forall u \in \mathcal{U}_{m}$, and removing those symbols with zero probability gives the updated alphabet $\mathcal{U}_{\tilde{m}}$ with $\tilde{m} \leq m - 1$. Define a new encoder $\tilde{K}$ on $\mathcal{U}_{\tilde{m}}$ as
\[
\tilde{K}(u|x) = \frac{\tilde{\lambda}_{u} w_{u}(x)}{P_{X}(x)}, \quad \forall u \in \mathcal{U}_{\tilde{m}}.
\]
It satisfies $I_{\tilde{K}}(U;X) = I_{K}(U;X) \leq R$, and $B_{\tilde{m}}(\tilde{K}, s_{K}) \leq B_{m}(K, s_{K})$. Hence,
\[
\mathcal{C}(\tilde{K}) = - \log \min_{0 \leq s \leq 1} B_{\tilde{m}}(\tilde{K}, s) \geq - \log B_{\tilde{m}}(\tilde{K}, s_{K}) \geq - \log B_{m}(K, s_{K}) = \mathcal{C}(K).
\]

Therefore, for any encoder $K$ on $\mathcal{U}_{m}$ with $m \geq |\mathcal{X}| + 2$, we can always construct a new encoder $\tilde{K}$, whose positive probability output symbol is reduced by at least one, and
\[
I_{\tilde{K}}(U;X) = I_{K}(U;X), \quad \mathcal{C}(\tilde{K}) \geq \mathcal{C}(K).
\]
Accordingly,
\[
\max_{K \in \mathcal{K}^{(|\mathcal{X}|+1)}_{R}} \mathcal{C}(K) \geq \max_{K \in \mathcal{K}^{(|\mathcal{X}|+2)}_{R}} \mathcal{C}(K) \geq \cdots \geq \max_{K \in \mathcal{K}^{(|\mathcal{X}|+n)}_{R}} \mathcal{C}(K), \quad n \to \infty.
\]
And since
\[
\mathscr{K}_{R} = \bigcup_{M \geq 1} \mathcal{K}^{(M)}_{R},
\]
we have
\[
\sup_{K \in \mathscr{K}_{R}} \mathcal{C}(K) = \max_{K \in \mathcal{K}^{(|\mathcal{X}|+1)}_{R}} \mathcal{C}(K).
\]

Furthermore, by Lemma \ref{lem: K_R non empty compact convex}, $\mathcal{K}^{(|\mathcal{X}|+1)}_{R}$ is non-empty and compact, and since $\mathcal{C}(K)$ is continuous with respect to $K$, the maximum in \eqref{eq: sup = max X+1} is attainable.
\end{proof}

Theorem \ref{the: supp U leq X+1} demonstrates the supremum in the entire feasible set $\mathscr{K}_{R}$ can be achieved in a feasible encoder set $\mathcal{K}^{(M)}_{R}$ with cardinality $M = |\mathcal{X}|+1$. In fact, the cardinality bound can be further sharpened, as Theorem \ref{the: sharpened cardinality} shows, while Lemma \ref{lem: reduction to V} is applied in the proof. Let $V$ be defined in \eqref{eq: definition V = nu(X)}, taking $k$ distinct values.

\begin{lemma}[Reduction to the sufficient statistic]
\label{lem: reduction to V}
    Consider an encoder $\bar{K} = \bar{K}_{U|V}$ of $V$ with output alphabet $\mathcal{U}_{M}$, denote by $P^{\bar{K}}_{U|y}(u) = \sum_{j=1}^{k} P_{V|y}(j) \bar{K}(u|j)$, $y \in \{0,1\}$, its induced conditional distributions, and by $I_{\bar{K}}(U;V)$ its mutual information under $P_{V}$. Then,
    \begin{enumerate}[label=(\roman*)]
        \item For every encoder $K = K_{U|X}$ with output alphabet $\mathcal{U}_{M}$, there exists an encoder $\bar{K}$ of $V$ with output alphabet $\mathcal{U}_{M}$ such that
        \begin{equation}
        \label{eq: reduction to V conclusion}
            P^{\bar{K}}_{U|y} = P^{K}_{U|y}, \ y \in \{0,1\}, \qquad I_{\bar{K}}(U;V) \leq I_{K}(U;X).
        \end{equation}
        \item For every encoder $\bar{K}$ of $V$ with output alphabet $\mathcal{U}_{M}$, the encoder $K(u|x) = \bar{K}(u|\nu(x))$ of $X$ satisfies
        \begin{equation}
        \label{eq: reduction to V converse}
            P^{K}_{U|y} = P^{\bar{K}}_{U|y}, \ y \in \{0,1\}, \qquad I_{K}(U;X) = I_{\bar{K}}(U;V).
        \end{equation}
    \end{enumerate}
\end{lemma}

\begin{proof}
    Appendix~\ref{proof: reduction to V}.
\end{proof}

\begin{theorem}[Sharpened cardinality bound]
\label{the: sharpened cardinality}
    For any $R \geq 0$,
    \begin{equation}
    \label{eq: sup = max k+1}
        \sup_{K \in \mathscr{K}_{R}} C(P^{K}_{U|0}, P^{K}_{U|1}) = \max_{K \in \mathcal{K}^{(k+1)}_{R}} C(P^{K}_{U|0}, P^{K}_{U|1}),
    \end{equation}
    and the maximum is attained by an encoder of the form $K(u|x) = \bar{K}(u|\nu(x))$ for some encoder $\bar{K} = \bar{K}_{U|V}$ of $V$ with output alphabet $\mathcal{U}_{k+1}$. Hence the cardinality of $U$ can be restricted to $|\mathcal{U}| \leq k+1$ without loss of generality.
\end{theorem}

\begin{proof}
    Consider the pair $(V,Y)$ with joint distribution
    \[
    P_{V,Y}(j,y) = \sum_{x \in \mathcal{X}_{j}} P_{X,Y}(x,y), \quad (j,y) \in \{1, \cdots, k\} \times \{0,1\}.
    \]
    The marginals are $P_{Y}(y) > 0$ and $P_{V}(j) = \sum_{x \in \mathcal{X}_{j}} P_{X}(x) > 0$, and by \eqref{eq: condition no instantly recognizable x}, $P_{V|y}(j) = P_{X|y}(\mathcal{X}_{j}) > 0$ for all $j$ and $y$. Moreover, summing $P_{X|1}(x) = \rho_{j} P_{X|0}(x)$ over $\mathcal{X}_{j}$ derives $P_{V|1}(j) = \rho_{j} P_{V|0}(j)$. Thus, $P_{X|0} \not\equiv P_{X|1}$ indicates $P_{V|0} \not\equiv P_{V|1}$.

    With the notation of Lemma \ref{lem: reduction to V}, and similar to the definition of $\mathscr{K}_{R}$ in \eqref{eq: union feasible encoder set}, denote by the entire feasible set from $V$ to $U$ as $\bar{\mathscr{K}}_{R}$, and define the curve of the reduced problem as
    \begin{equation}
    \label{eq: definition C_V(R)}
    C_{V}(R) \stackrel{\Delta}{=} \sup_{\bar{K} \in \bar{\mathscr{K}}_{R}} C(P^{\bar{K}}_{U|0}, P^{\bar{K}}_{U|1}).
    \end{equation}
    Since the proof of Theorem \ref{the: supp U leq X+1} uses only the standing assumptions of Section \ref{sec: Setup}, it applies verbatim to the pair $(V,Y)$, whose alphabet has cardinality $k$. Therefore, there exists an encoder $\bar{K}^{*}$ of $V$ with output alphabet $\mathcal{U}_{k+1}$ such that
    \begin{equation}
    \label{eq: reduced problem optimum}
        I_{\bar{K}^{*}}(U;V) \leq R, \quad C \big(P^{\bar{K}^{*}}_{U|0}, P^{\bar{K}^{*}}_{U|1} \big) = C_{V}(R).
    \end{equation}

    We now show that $C(R) = C_{V}(R)$. Let $K \in \mathscr{K}_{R}$ be arbitrary, with output alphabet $\mathcal{U}_{M}$, and let $\bar{K}$ be defined by Lemma \ref{lem: reduction to V}(i). By \eqref{eq: reduction to V conclusion}, $I_{\bar{K}}(U;V) \leq I_{K}(U;X) \leq R$ and $P^{\bar{K}}_{U|y} = P^{K}_{U|y}$ for $y \in \{0,1\}$. Hence
    \[
    \mathcal{C}(K) = C \big(P^{K}_{U|0}, P^{K}_{U|1} \big) = C \big(P^{\bar{K}}_{U|0}, P^{\bar{K}}_{U|1} \big) \leq C_{V}(R),
    \]
    and taking the supremum over $K \in \mathscr{K}_{R}$ gives $C(R) \leq C_{V}(R)$. Conversely, let $K^{*}(u|x) = \bar{K}^{*}(u|\nu(x))$. By \eqref{eq: reduction to V converse} and \eqref{eq: reduced problem optimum}, $K^{*}$ has output alphabet $\mathcal{U}_{k+1}$, $I_{K^{*}}(U;X) = I_{\bar{K}^{*}}(U;V) \leq R$, and $P^{K^{*}}_{U|y} = P^{\bar{K}^{*}}_{U|y}$ for $y \in \{0,1\}$, so that
    \[
    K^{*} \in \mathcal{K}^{(k+1)}_{R}, \quad \mathcal{C}(K^{*}) = C_{V}(R) \geq C(R).
    \]

    Combining the two directions, and using $\mathcal{K}^{(k+1)}_{R} \subseteq \mathscr{K}_{R}$,
    \[
    C(R) \leq \mathcal{C}(K^{*}) \leq \max_{K \in \mathcal{K}^{(k+1)}_{R}} \mathcal{C}(K) \leq \sup_{K \in \mathscr{K}_{R}} \mathcal{C}(K) = C(R).
    \]
    Hence all inequalities hold with equality, and the supremum over $\mathcal{K}^{(k+1)}_{R}$ is attained by $K^{*}$, which proves \eqref{eq: sup = max k+1} with a maximizer of the form $K^{*}(u|x) = \bar{K}^{*}(u|\nu(x))$.
\end{proof}

\begin{remark}
    Theorem \ref{the: sharpened cardinality} strictly improves Theorem \ref{the: supp U leq X+1} when $k < |\mathcal{X}|$, i.e., at least two symbols share the same likelihood ratio. Lemma \ref{lem: reduction to V} also shows that the curve $C(R)$ depends on the source only through the joint distribution $P_{V,Y}$: compressing $X$ is never better than compressing its minimal sufficient statistic, at any rate.
\end{remark}

Therefore, we can reformulate the MI constrained Chernoff bottleneck as follows:

\begin{problem}
\label{prob: reformulated problem}
Given $R \geq 0$ and $\mathcal{U}_{M} = \{1,\cdots,M\}$, seek
\begin{equation}
\label{eq: problem 2}
    K^{*} \in \arg \max_{K \in \mathcal{K}^{(M)}_{R}} C(P^{K}_{U|0}, P^{K}_{U|1}).
\end{equation}
\end{problem}

Theorem \ref{the: sharpened cardinality} implies that we can recover the optimal solution of Problem \ref{prob: initial problem} by setting $M \geq k + 1$ in Problem \ref{prob: reformulated problem}, for which we establish an alternating optimization algorithm in Section \ref{sec: An Iterative Algorithm}.

\section{An Alternating Optimization Algorithm}\label{sec: An Iterative Algorithm}

In this section, we fix the alphabet $\mathcal{U}_{M} = \{1,\cdots,M\}$. With $B_{M}(K, s)$ defined in \eqref{eq: definition BM(K,s)}, the maximization in \eqref{eq: problem 2} can be rewritten as
\[
\min_{K \in \mathcal{K}^{(M)}_{R}}  \min_{0 \leq s \leq 1} \ B_{M}(K, s).
\]
We establish an iterative algorithm for Problem \ref{prob: reformulated problem} by alternately minimizing over $K$ and $s$:
\begin{enumerate}[label=(\alph*)]
    \item\label{item: algorithm step 1} Fix $s$, update $K$ with
    \begin{equation}
    \label{eq: fix s update K with}
        K \in \arg \min_{K \in \mathcal{K}^{(M)}_{R}} \ B_{M}(K, s);
    \end{equation}
    \item\label{item: algorithm step 2} Fix $K$, update $s$ with
    \begin{equation}
    \label{eq: fix K update s with}
        s = \arg \min_{0 \leq s \leq 1} \ B_{M}(K, s).
    \end{equation}
\end{enumerate}

\subsection{Step \ref{item: algorithm step 1}: Generalized Blahut-Arimoto Algorithm}

In this step, we fix $s$ and update $K$. For fixed $s \in (0,1)$, and any $w(u) > 0$, Young's inequality (Fact \ref{fact: Young's inequality}) gives
\[
P^{K}_{U|0}(u)^{s} P^{K}_{U|1}(u)^{1-s} \leq s w(u) P^{K}_{U|0}(u) + (1-s) w(u)^{-\frac{s}{1-s}} P^{K}_{U|1}(u),
\]
and equality holds when $w(u) P^{K}_{U|0}(u) = w(u)^{-\frac{s}{1-s}} P^{K}_{U|1}(u)$. Namely,
\[
P^{K}_{U|0}(u)^{s} P^{K}_{U|1}(u)^{1-s} = \min_{w(u) > 0} \big(s w(u) P^{K}_{U|0}(u) + (1-s) w(u)^{-\frac{s}{1-s}} P^{K}_{U|1}(u) \big),
\]
with the minimum achieved by
\begin{equation}
\label{eq: w(u) achieve minimum}
    w(u) = \Big(\frac{P^{K}_{U|1}(u)}{P^{K}_{U|0}(u)} \Big)^{1-s}.
\end{equation}

\begin{lemma}
\label{lem: uniform bound of w(u)}
    Let $\mathcal{K}^{(M)}_{R,+}$ be defined in \eqref{eq: positive feasible encoder set K(M)R+}, and $\ell$ and $L$ be defined in \eqref{eq: definition of ell and L}. If $K_{U|X} \in \mathcal{K}^{(M)}_{R,+}$, then for any $u \in \mathcal{U}_{M}$,
    \[
    0 < \ell \leq \frac{P^{K}_{U|1}(u)}{P^{K}_{U|0}(u)} \leq L < +\infty.
    \]
    Hence, $\ell^{1-s} \leq w(u) \leq L^{1-s}$.
\end{lemma}

\begin{proof}
    Appendix~\ref{proof: uniform bound of w(u)}.
\end{proof}

\begin{remark}
    Lemma \ref{lem: uniform bound of w(u)} gives a uniform bound of $w(u)$. The condition $K_{U|X} \in \mathcal{K}^{(M)}_{R,+}$ holds for every iterate of Algorithm \ref{alg: Algorithm for reformulated problem} by Theorem \ref{the: convergence of algorithm_reformulated problem}.
\end{remark}

Thus, the minimization in \eqref{eq: fix s update K with} becomes
\begin{equation}
\label{eq: minimizing K s w}
    \min_{K \in \mathcal{K}^{(M)}_{R}} \min_{w(u) > 0} \ \sum_{u \in \mathcal{U}_{M}} s w(u) P^{K}_{U|0}(u) + (1-s) w(u)^{-\frac{s}{1-s}} P^{K}_{U|1}(u).
\end{equation}
For $y \in \{0,1\}$, Markovity gives
\[
P^{K}_{U|y}(u) = \sum_{x \in \mathcal{X}} K(u|x) P_{X|y}(x).
\]
Define the parameter matrix $C^{s} \in \mathbb{R}_{+}^{M \times |\mathcal{X}|}$ with entries
\begin{equation}
\label{eq: parameter matrix Cs}
    C^{s}_{u,x} = s w(u) P_{X|0}(x) + (1-s) w(u)^{-\frac{s}{1-s}} P_{X|1}(x),
\end{equation}
where $w(u)$ is given by \eqref{eq: w(u) achieve minimum} at the current encoder. With $w(u)$ fixed, the objective in \eqref{eq: minimizing K s w} is linear in $K$, and the update of $K$ solves
\begin{equation}
\label{eq: sum u,x C K}
    \min_{K \in \mathcal{K}^{(M)}_{R}} \ \sum_{(u,x) \in \mathcal{U}_{M} \times \mathcal{X}} C^{s}_{u,x} K(u|x).
\end{equation}

\begin{remark}[Concavity and the linear majorizer]
\label{rem: concavity of B}
    For fixed $s \in [0,1]$, $B_{M}(K,s)$ is concave in $K$, since $(a,b) \mapsto a^{s} b^{1-s}$ is concave on $\mathbb{R}_{+}^{2}$ and $P^{K}_{U|y}$ is linear in $K$. Consequently, maximizing $\mathcal{C}(K)$ amounts to minimizing the concave function $\min_{s \in [0,1]} B_{M}(K,s)$ over the convex set $\mathcal{K}^{(M)}_{R}$, for which global optimality cannot be expected in general. Moreover, let $w(u)$ and $C^{s}_{u,x}$ be computed from a current encoder $\tilde{K} \in \mathcal{K}^{(M)}_{R,+}$ via \eqref{eq: w(u) achieve minimum} and \eqref{eq: parameter matrix Cs}. For $s \in (0,1)$, direct differentiation gives
    \[
    \frac{\partial B_{M}(K,s)}{\partial K(u|x)} \bigg|_{K = \tilde{K}} = s w(u) P_{X|0}(x) + (1-s) w(u)^{-\frac{s}{1-s}} P_{X|1}(x) = C^{s}_{u,x}.
    \]
    Since $B_{M}(\cdot,s)$ is positively homogeneous of degree one, Euler's identity gives
    \[
    B_{M}(\tilde{K},s) = \sum_{u,x} C^{s}_{u,x} \tilde{K}(u|x).
    \]
    Hence the objective in \eqref{eq: sum u,x C K} is exactly the tangent plane of $B_{M}(\cdot,s)$ at $\tilde{K}$: by concavity, it upper bounds $B_{M}(K,s)$ for every $K$, with equality at $K = \tilde{K}$. Step \ref{item: algorithm step 1} is therefore a successive linearization of a concave function.
\end{remark}

Write the Lagrangian function as
\[
\mathcal{L}(K, \lambda, \mu) = \sum_{u,x} C^{s}_{u,x} K(u|x) + \lambda \big(I_{K}(U;X) - R \big) + \sum_{x} \mu(x) \big(\sum_{u} K(u|x) - 1 \big),
\]
with the trade-off parameter $\lambda \geq 0$. Taking derivative over $K$ gives
\[
\frac{\partial \mathcal{L}}{\partial K} = C^{s}_{u,x} + \lambda P_{X}(x) \log \frac{K(u|x)}{P_{U}(u)} + \mu(x).
\]
Let $\frac{\partial \mathcal{L}}{\partial K} = 0$ and normalize, we have
\begin{equation}
\label{eq: fix PU update KU|X}
    K(u|x) = \frac{P_{U}(u) \text{exp}\Big( -\dfrac{C^{s}_{u,x}}{\lambda P_{X}(x)} \Big)}{ \sum_{u'} P_{U}(u') \text{exp}\Big( -\dfrac{C^{s}_{u',x}}{\lambda P_{X}(x)} \Big)}.
\end{equation}
And for fixed $K(u|x)$, the output distribution $P_{U}(u)$ is given by $P_{U}(u) = \sum_{x} K(u|x) P_{X}(x)$. Thus, we have established the Generalized Blahut-Arimoto (GBA) algorithm for \eqref{eq: sum u,x C K}.

\begin{algorithm}[H]
\caption{Generalized Blahut-Arimoto Algorithm for \eqref{eq: sum u,x C K}}
\label{alg: GBA for K_U|X}
\begin{algorithmic}[1]
\Require $P_{X}$; parameter matrix $C^{s}_{u,x}$; current $K_{U|X}$; trade-off $\lambda > 0$
\Ensure converged encoder $\hat{K}_{U|X}$
\State $K_{U|X}^{(0)} \leftarrow K_{U|X}$; \ $P_{U}^{(0)}(u) \leftarrow \sum_{x} K_{U|X}^{(0)}(u|x) P_{X}(x)$; \ $t \leftarrow 0$
\Repeat
    \State Given $P_{U}^{(t)}$, update $K_{U|X}^{(t+1)}$:
    \[
    K^{(t+1)}(u|x) \leftarrow \frac{P_{U}^{(t)}(u) \exp\Big( -\dfrac{C^{s}_{u,x}}{\lambda P_{X}(x)} \Big)}{ \sum_{u'} P_{U}^{(t)}(u') \exp\Big( -\dfrac{C^{s}_{u',x}}{\lambda P_{X}(x)} \Big)}
    \]

    \State Given $K_{U|X}^{(t+1)}$, update $P_{U}^{(t+1)}$:
    \[
    P_{U}^{(t+1)}(u) \leftarrow \sum_{x} K^{(t+1)}(u|x) P_{X}(x)
    \]
    
    \State $t \leftarrow t+1$
\Until{ $\max_{u \in \mathcal{U}_{M}} \big( P^{(t)}_{U}(u) - P^{(t-1)}_{U}(u) \big) \big/ P^{(t-1)}_{U}(u)$ very small} \Comment{as stated in Rem.\ \ref{rem: stopping criterion GBA}}
\State \Return $\hat{K}_{U|X} = K_{U|X}^{(t)}$
\end{algorithmic}
\end{algorithm}

\begin{remark}
\label{rem: stopping criterion GBA}
    The stopping criterion is that of the Blahut--Arimoto algorithm \cite{blahut1972,arimoto1972}, and it comes with an optimality certificate. Let $\epsilon_{t} = \max_{u} \big( P^{(t)}_{U}(u) - P^{(t-1)}_{U}(u) \big) / P^{(t-1)}_{U}(u) \geq 0$ be the quantity checked at termination. By Corollary \ref{cor: GBA stopping certificate}, the output $\hat{K}_{U|X} = K^{(t)}_{U|X}$ satisfies
    \[
    \mathcal{L}_{\lambda}(\hat{K}) - \min_{K} \mathcal{L}_{\lambda}(K) \leq \lambda \log(1 + \epsilon_{t}) \leq \lambda \epsilon_{t},
    \]
    where $\mathcal{L}_{\lambda}$ is defined in Section \ref{sec: analysis}.
\end{remark}

\subsection{Step \ref{item: algorithm step 2}: Solving a Nonlinear Equation}

In this step, we fix $K$ and update $s$. We consider the encoder $K \in \mathcal{K}^{(M)}_{R,+}$, which holds for every iterate of Algorithm \ref{alg: Algorithm for reformulated problem} by Theorem \ref{the: convergence of algorithm_reformulated problem}.

For fixed $\mathcal{U}_{M}$ and $K_{U|X} \in \mathcal{K}^{(M)}_{R,+}$, define
\begin{equation}
\label{eq: definition of g_{K}(s)}
    g_{K}(s) \stackrel{\Delta}{=} \sum_{u \in \mathcal{U}_{M}} P^{K}_{U|0}(u)^{s} P^{K}_{U|1}(u)^{1-s}, \quad s \in [0,1]
\end{equation}

\begin{lemma}
\label{lem: gKs continuous and infinitely differentiable}
    For $K_{U|X} \in \mathcal{K}^{(M)}_{R,+}$, $g_{K}(s)$ is continuous and infinitely differentiable on $[0,1]$.
\end{lemma}

\begin{proof}
$K_{U|X} \in \mathcal{K}^{(M)}_{R,+}$ and the Markovity guarantees for any $u \in \mathcal{U}_{M}$, $P^{K}_{U|0}(u) > 0, P^{K}_{U|1}(u) > 0$. Since $g_{K}(s)$ is a summation of finite terms, considering $h_{K}^{u}(s) \stackrel{\Delta}{=} P^{K}_{U|0}(u)^{s} P^{K}_{U|1}(u)^{1-s}$, we have
\[
h_{K}^{u}(s) = \exp \big(s \log P^{K}_{U|0}(u) + (1-s) \log P^{K}_{U|1}(u) \big),
\]
which is certainly continuous and infinitely differentiable.
\end{proof}

Calculate the derivatives
\begin{equation}
\label{eq: derivative g_K'(s)}
    g_{K}'(s) = \sum_{u \in \mathcal{U}_{M}} P^{K}_{U|0}(u)^{s} P^{K}_{U|1}(u)^{1-s} \log \frac{P^{K}_{U|0}(u)}{P^{K}_{U|1}(u)},
\end{equation}
and
\begin{equation}
\label{eq: derivative g_K''(s)}
    g_{K}''(s) = \sum_{u \in \mathcal{U}_{M}} P^{K}_{U|0}(u)^{s} P^{K}_{U|1}(u)^{1-s} \Big(\log \frac{P^{K}_{U|0}(u)}{P^{K}_{U|1}(u)} \Big)^2.
\end{equation}
Thus, $g_{K}''(s) \geq 0$ on $[0,1]$. Moreover, if $P^{K}_{U|0} \not\equiv P^{K}_{U|1}$ on $\mathcal{U}_{M}$, we have $g_{K}''(s) > 0$ on $[0,1]$, i.e., $g_{K}'(s)$ is monotonically increasing on $[0,1]$, and
\begin{equation}
\label{eq: g_K'(0) < 0}
    g_{K}'(0) = \sum_{u \in \mathcal{U}_{M}} P^{K}_{U|1}(u) \log \frac{P^{K}_{U|0}(u)}{P^{K}_{U|1}(u)} = - \mathrm{D}(P^{K}_{U|1} \| P^{K}_{U|0}) < 0,
\end{equation}
\begin{equation}
\label{eq: g_K'(1) > 0}
    g_{K}'(1) = \sum_{u \in \mathcal{U}_{M}} P^{K}_{U|0}(u) \log \frac{P^{K}_{U|0}(u)}{P^{K}_{U|1}(u)} = \mathrm{D}(P^{K}_{U|0} \| P^{K}_{U|1}) > 0.
\end{equation}
Hence, we have the following lemma.

\begin{lemma}
\label{lem: unique S^* achieve minimum}
    For $K_{U|X} \in \mathcal{K}^{(M)}_{R,+}$, $P^{K}_{U|0} \not\equiv P^{K}_{U|1}$, there exists the unique $s^{*} \in (0,1)$ such that
    \[
    g_{K}(s^{*}) = \min_{s \in [0,1]} g_{K}(s).
    \]
\end{lemma}

\begin{proof}
    The conditions give \eqref{eq: g_K'(0) < 0} and \eqref{eq: g_K'(1) > 0}. And since $g_{K}''(s) > 0$ on $[0,1]$, there exists the unique $s^{*} \in (0,1)$ such that $g_{K}'(s^{*}) = 0$.
\end{proof}

Searching $s^{*} = \arg \min_{s \in [0,1]} g_{K}(s)$ is equivalent to searching $s^{*}$ such that
\begin{equation}
\label{eq: g'_K(s^*) = 0}
    g_{K}'(s^{*}) = 0,
\end{equation}
where $g_{K}'(s)$ is given by \eqref{eq: derivative g_K'(s)}. Non-linear equation \eqref{eq: g'_K(s^*) = 0} can be solved by numerical methods such as the bisection method, Newton's method, or Brent's method \cite{burden2011, brent1973, press2007}.

\subsection{Alternating Optimization Algorithm for Problem \ref{prob: reformulated problem}}

Combining step \ref{item: algorithm step 1} and \ref{item: algorithm step 2}, we establish the algorithm for Problem \ref{prob: reformulated problem} as below.

\begin{algorithm}[H]
\caption{Algorithm for Problem \ref{prob: reformulated problem}}
\label{alg: Algorithm for reformulated problem}
\begin{algorithmic}[1]
\Require $P_{X}$; $P_{X|0}$; $P_{X|1}$; cardinality $M$; rate $R > 0$
\Ensure Encoder $K^{*}$
\State \textbf{Initialize:} $K_{U|X}^{(0)} \in \mathcal{K}^{(M)}_{R,+}$; \ $s^{(0)} \in (0,1)$; \ $l \leftarrow 0$
\State Compute $P_{U}^{(0)}$, $P_{U|0}^{(0)}$, $P_{U|1}^{(0)}$ from $K_{U|X}^{(0)}$; \ $B^{(0)} \leftarrow \sum_{u \in \mathcal{U}_{M}} P_{U|0}^{(0)}(u)^{s^{(0)}} P_{U|1}^{(0)}(u)^{1-s^{(0)}}$
\State \textbf{if} $B^{(0)} = 1$ \textbf{then go to} line 1 \Comment{re-initialize to avoid $P_{U|0}^{(0)} \equiv P_{U|1}^{(0)}$}
\State Compute $\Gamma$ by \eqref{eq: definition of Gamma}
\Repeat
    \State For each $u, x$, compute $w(u)$ from \eqref{eq: w(u) achieve minimum} and $C^{s^{(l)}}_{u,x}$ from \eqref{eq: parameter matrix Cs}
    \State $\lambda_{\text{small}} \leftarrow 0$; \ $\lambda_{\text{big}} \leftarrow \Gamma / R$; \ $K_{U|X}^{(l+1)} \leftarrow K_{U|X}^{(l)}$ \Comment{initial bracket (Lem.\ \ref{lem: bisection search of lambda})}
    \Repeat
        \State $\lambda \leftarrow (\lambda_{\text{small}} + \lambda_{\text{big}}) / 2$ \Comment{bisection (Lem.\ \ref{lem: bisection search of lambda})}
        \State Start from $K_{U|X}^{(l)}$, run Algorithm \ref{alg: GBA for K_U|X} with $\lambda$ to derive $\hat{K}_{U|X}$; compute $\hat{I}_{K}(U;X)$
        \If{$\hat{I}_{K}(U;X) > R$}
            \State $\lambda_{\text{small}} \leftarrow \lambda$
        \Else
            \State $\lambda_{\text{big}} \leftarrow \lambda$
            \If{$\sum_{u,x} C^{s^{(l)}}_{u,x} \hat{K}(u|x) \leq \sum_{u,x} C^{s^{(l)}}_{u,x} K^{(l+1)}(u|x)$} \Comment{safeguard}
                \State $K_{U|X}^{(l+1)} \leftarrow \hat{K}_{U|X}$
            \EndIf
        \EndIf
    \Until{$\lambda_{\text{big}} - \lambda_{\text{small}}$ very small}
    \State Compute $P_{U}^{(l+1)}$, $P_{U|0}^{(l+1)}$, $P_{U|1}^{(l+1)}$ from $K_{U|X}^{(l+1)}$
    \State Compute $g_{K}'(s)$ from \eqref{eq: derivative g_K'(s)}; solve equation \eqref{eq: g'_K(s^*) = 0} to derive $s^{(l+1)}$
    \State $B^{(l+1)} \leftarrow \sum_{u \in \mathcal{U}_{M}} P_{U|0}^{(l+1)}(u)^{s^{(l+1)}} P_{U|1}^{(l+1)}(u)^{1-s^{(l+1)}}$ \Comment{update objective $B$}
    \State $l \leftarrow l+1$
\Until{$|B^{(l)} - B^{(l-1)}|$ very small}
\State \Return $K^{*} = K_{U|X}^{(l)}$
\end{algorithmic}
\end{algorithm}

\subsection{Convergence Analysis}\label{sec: analysis}

We analyze the convergence of Algorithms \ref{alg: GBA for K_U|X} and \ref{alg: Algorithm for reformulated problem}.

\begin{lemma}[Uniform boundedness]
\label{lem: algorithm GBA uniform bounded}
    For $K_{U|X} \in \mathcal{K}^{(M)}_{R,+}$ and $s \in (0,1)$, let $w(u)$ be given by \eqref{eq: w(u) achieve minimum} and $C^{s}_{u,x}$ by \eqref{eq: parameter matrix Cs}. Then, for all $(u,x) \in \mathcal{U}_{M} \times \mathcal{X}$,
    \[
    0 < P_{X|0}(x)^{s}P_{X|1}(x)^{1-s} \leq C^{s}_{u,x} \leq \Gamma \big(sP_{X|0}(x) +(1-s) P_{X|1}(x) \big).
    \]
\end{lemma}

\begin{proof}
    Appendix~\ref{proof: algorithm GBA uniform bounded}.
\end{proof}

For fixed $C^{s}_{u,x}$ and $\lambda \geq 0$, denote by $\mathcal{K}^{(M)}_{\infty}$ the set of all encoders with output alphabet $\mathcal{U}_{M}$, with no rate constraint, and define
\[
G(K) \stackrel{\Delta}{=} \sum_{u,x} C^{s}_{u,x} K(u|x), \quad \mathcal{L}_{\lambda}(K) \stackrel{\Delta}{=} G(K) + \lambda I_{K}(U;X),
\]
so that \eqref{eq: sum u,x C K} reads $\min_{K \in \mathcal{K}^{(M)}_{R}} G(K)$. Let $\mathcal{L}^{*}_{\lambda} = \min_{K \in \mathcal{K}^{(M)}_{\infty}} \mathcal{L}_{\lambda}(K)$, which is attained since $\mathcal{K}^{(M)}_{\infty}$ is compact and $\mathcal{L}_{\lambda}$ is continuous. For $\lambda > 0$, Algorithm \ref{alg: GBA for K_U|X} alternately minimizes
\begin{equation}
\label{eq: definition hatL}
    \hat{\mathcal{L}}_{\lambda}(K_{U|X},P_U) = \sum_{u,x} C^{s}_{u,x} K(u|x) + \lambda \sum_{x} P_{X}(x) \sum_{u} K(u|x) \log \frac{K(u|x)}{P_U(u)}
\end{equation}
over $K_{U|X} \in \mathcal{K}^{(M)}_{\infty}$ and $P_{U}$. Since $\sum_{x} P_{X}(x) \sum_{u} K(u|x) \log \frac{K(u|x)}{P_{U}(u)} = I_{K}(U;X) + \mathrm{D}(P^{K}_{U} \| P_{U})$,
\begin{equation}
\label{eq: hatL = L + D}
    \hat{\mathcal{L}}_{\lambda}(K,P_U) = \mathcal{L}_{\lambda}(K) + \lambda \mathrm{D}(P^{K}_{U} \| P_{U}),
\end{equation}
so that $\min_{P_U} \hat{\mathcal{L}}_{\lambda}(K,P_U) = \mathcal{L}_{\lambda}(K)$, attained at $P_{U} = P^{K}_{U}$, and the global minimum of $\hat{\mathcal{L}}_{\lambda}$ equals $\mathcal{L}^{*}_{\lambda}$.

\begin{theorem}[Global convergence of Algorithm \ref{alg: GBA for K_U|X}]
\label{the: hatL global converge}
    Given $C^{s}_{u,x}$ and $\lambda > 0$, assume $P^{(0)}_{U}(u) > 0$ for all $u \in \mathcal{U}_{M}$ in Algorithm \ref{alg: GBA for K_U|X}, which holds if $K^{(0)}_{U|X} \in \mathcal{K}^{(M)}_{R,+}$. Then $K^{(t)}_{U|X}$ converges to a global minimizer $\bar{K}$ of $\mathcal{L}_{\lambda}$ over $\mathcal{K}^{(M)}_{\infty}$, and $(K^{(t)}_{U|X}, P^{(t)}_{U})$ converges to $(\bar{K}, P^{\bar{K}}_{U})$, a global minimizer of $\hat{\mathcal{L}}_{\lambda}$. Moreover, $\mathcal{L}_{\lambda}(K^{(t)}_{U|X})$ is nonincreasing in $t$, and for every global minimizer $K^{*}$ of $\mathcal{L}_{\lambda}$ over $\mathcal{K}^{(M)}_{\infty}$ and every $t \geq 1$,
    \begin{equation}
    \label{eq: GBA rate}
        0 \leq \mathcal{L}_{\lambda}(K^{(t)}_{U|X}) - \mathcal{L}^{*}_{\lambda} \leq \frac{\lambda \, \mathrm{D}(P^{K^{*}}_{U} \| P^{(0)}_{U})}{t}.
    \end{equation}
\end{theorem}

\begin{proof}
    Appendix~\ref{proof: hatL global converge}.
\end{proof}

\begin{corollary}[Stopping certificate]
\label{cor: GBA stopping certificate}
    Under the assumptions of Theorem \ref{the: hatL global converge}, for every $t \geq 1$,
    \[
    0 \leq \mathcal{L}_{\lambda}(K^{(t)}_{U|X}) - \mathcal{L}^{*}_{\lambda} \leq \lambda \log \max_{u \in \mathcal{U}_{M}} \frac{P^{(t)}_{U}(u)}{P^{(t-1)}_{U}(u)}.
    \]
\end{corollary}

\begin{proof}
    Let $K^{*}$ be a global minimizer of $\mathcal{L}_{\lambda}$ and $q^{*} = P^{K^{*}}_{U}$. Applying \eqref{eq: GBA one-step} at step $t-1$ with $K = K^{*}$, in the notation of Appendix \ref{proof: hatL global converge},
    \[
    \mathcal{L}_{\lambda}(K_{t}) - \mathcal{L}^{*}_{\lambda} \leq \lambda \big( \mathrm{D}(q^{*} \| q_{t-1}) - \mathrm{D}(q^{*} \| q_{t}) \big) = \lambda \sum_{u} q^{*}(u) \log \frac{q_{t}(u)}{q_{t-1}(u)} \leq \lambda \log \max_{u} \frac{q_{t}(u)}{q_{t-1}(u)}.
    \]
\end{proof}

In Algorithm \ref{alg: Algorithm for reformulated problem}, the trade-off parameter $\lambda$ is searched by bisection. To justify this, let $K(\lambda)$ be any minimizer of $\mathcal{L}_{\lambda}(K)$ in $\mathcal{K}^{(M)}_{\infty}$, i.e.,
\begin{equation}
\label{eq: definition K(lambda)}
    K(\lambda) \in \arg \min_{K \in \mathcal{K}^{(M)}_{\infty}} \mathcal{L}_{\lambda}(K),
\end{equation}
which Algorithm \ref{alg: GBA for K_U|X} approximates by Theorem \ref{the: hatL global converge}, and let $I(\lambda) \stackrel{\Delta}{=} I_{K(\lambda)}(U;X)$. The minimizer in \eqref{eq: definition K(lambda)} need not be unique, and the following result holds for every choice of it.

\begin{lemma}[Bisection search of $\lambda$]
\label{lem: bisection search of lambda}
    Let $C^{s}_{u,x}$ be given by \eqref{eq: parameter matrix Cs} for some $K_{U|X} \in \mathcal{K}^{(M)}_{R,+}$ and $s \in (0,1)$, let $R > 0$, $\Gamma$ be defined in \eqref{eq: definition of Gamma}, and $G^{*} = \min_{K \in \mathcal{K}^{(M)}_{R}} G(K)$ be the optimal value of \eqref{eq: sum u,x C K}. Then, for every choice of $K(\lambda)$ in \eqref{eq: definition K(lambda)}:
    \begin{enumerate}[label=(\roman*)]
        \item\label{item: bisection monotone} $I(\lambda)$ is nonincreasing on $[0, +\infty)$, i.e., $I(\lambda_1) \geq I(\lambda_2)$ for all $0 \leq \lambda_1 < \lambda_2$.
            
        \item\label{item: bisection certificate} If $I(\lambda) \leq R$, then $K(\lambda) \in \mathcal{K}^{(M)}_{R}$ and
        \begin{equation}
        \label{eq: bisection certificate}
            0 \leq G(K(\lambda)) - G^{*} \leq \lambda \big(R - I(\lambda) \big).
        \end{equation}
        In particular, $K(\lambda)$ solves \eqref{eq: sum u,x C K} if $I(\lambda) = R$ or $\lambda = 0$.
        
        \item\label{item: bisection threshold} There exists $\lambda^{*} \in [0, \Gamma/R)$ such that $I(\lambda) > R$ for $\lambda < \lambda^{*}$, and $I(\lambda) \leq R$ for $\lambda > \lambda^{*}$.
    \end{enumerate}
    Consequently, the bisection in lines 7--19 of Algorithm \ref{alg: Algorithm for reformulated problem}, initialized with $\lambda_{\text{small}} = 0$ and $\lambda_{\text{big}} = \Gamma / R$, keeps $\lambda^{*} \in [\lambda_{\text{small}}, \lambda_{\text{big}}]$ and $I(\lambda_{\text{big}}) \leq R$ at every step, and after $n$ steps, $\lambda_{\text{big}} - \lambda_{\text{small}} = \Gamma / (R \, 2^{n})$.
\end{lemma}

\begin{proof}
    Appendix~\ref{proof: bisection search of lambda}.
\end{proof}

\begin{remark}
    Lemma \ref{lem: bisection search of lambda} justifies lines 7--19 of Algorithm \ref{alg: Algorithm for reformulated problem}. The bracket $[0, \Gamma/R]$ is valid at every outer iteration, since the bound $\lambda^{*} < \Gamma/R$ relies only on Lemma \ref{lem: algorithm GBA uniform bounded}. Locating $\lambda^{*}$ to accuracy $\delta > 0$ takes $\lceil \log_{2} (\Gamma / (R \delta)) \rceil$ runs of Algorithm \ref{alg: GBA for K_U|X}. Finally, $K^{(l+1)}_{U|X}$ is only replaced by encoders of rate at most $R$, so it stays feasible for \eqref{eq: sum u,x C K}, and its suboptimality is bounded by \eqref{eq: bisection certificate}.
\end{remark}

Many previous theorems and lemmas assume $s \in (0,1)$; the following lemma ensures this.

\begin{lemma}[Interiority of $s^{*}$]
\label{lem: interiority of s^*}
    For $K_{U|X} \in \mathcal{K}^{(M)}_{R,+}$ with $P^{K}_{U|0} \not\equiv P^{K}_{U|1}$, let $\mathcal{C}(K)$ be defined in \eqref{eq: definition Chernoff CK} as the corresponding Chernoff information, $\Gamma$ be defined in \eqref{eq: definition of Gamma}, and $s^{*}$ be the unique root of \eqref{eq: g'_K(s^*) = 0}. Then,
    \begin{equation}
    \label{eq: interval of s^*}
        \frac{\mathcal{C}(K)}{\log \Gamma} \leq s^{*} \leq 1 - \frac{\mathcal{C}(K)}{\log \Gamma}.
    \end{equation}
    Moreover, $\mathcal{C}(K) \leq \frac{1}{2} \log \Gamma$, which indicates the interval \eqref{eq: interval of s^*} is non-empty.
\end{lemma}

\begin{proof}
By Lemma \ref{lem: uniform bound of w(u)}, for any $u \in \mathcal{U}_{M}$,
\[
\Big| \log \frac{P^{K}_{U|1}(u)}{P^{K}_{U|0}(u)} \Big| \leq \max \{\log L, \log \ell^{-1} \} = \log \Gamma.
\]
Thus,
\begin{equation}
\label{eq: D(P1 P0) D(P0 P1) leq log Gamma}
    D(P^{K}_{U|1} \| P^{K}_{U|0}) \leq \log \Gamma, \quad D(P^{K}_{U|0} \| P^{K}_{U|1}) \leq \log \Gamma.
\end{equation}

With $g_{K}(s)$ defined in \eqref{eq: definition of g_{K}(s)}, $K_{U|X} \in \mathcal{K}^{(M)}_{R,+}$ indicates $g_{K}(s) > 0$. Let $\psi_{K}(s) \stackrel{\Delta}{=} \log g_{K}(s)$. While Lemma \ref{lem: gKs continuous and infinitely differentiable} guarantees the differentiability, calculate the derivatives gives
\begin{equation}
\label{eq: derivatives psi_K'(s) psi_K''(s)}
    \psi_{K}'(s) = \frac{g_{K}'(s)}{g_{K}(s)}, \quad \psi_{K}''(s) = \frac{g_{K}''(s)g_{K}(s) - (g_{K}'(s))^2}{(g_{K}(s))^2}.
\end{equation}
By \eqref{eq: derivative g_K'(s)}, \eqref{eq: derivative g_K''(s)} and Cauchy-Schwarz inequality, we have $(g_{K}'(s))^2 \leq g_{K}''(s)g_{K}(s)$, and equality holds if and only if $P^{K}_{U|0} \equiv P^{K}_{U|1}$. Hence, $\psi_{K}(s)$ is strictly convex on $[0,1]$. Moreover, $g_{K}(0) = g_{K}(1) = 1$. Thus, $\psi_{K}(0) = \psi_{K}(1) = 0$, and by \eqref{eq: g_K'(0) < 0}, \eqref{eq: g_K'(1) > 0} and \eqref{eq: derivatives psi_K'(s) psi_K''(s)},
\[
\psi_{K}'(0) = g_{K}'(0) = -D(P^{K}_{U|1} \| P^{K}_{U|0}), \quad \psi_{K}'(1) = g_{K}'(1) = D(P^{K}_{U|0} \| P^{K}_{U|1}).
\]

Since $\arg \min_{s \in [0,1]} \psi_{K}(s) = \arg \min_{s \in [0,1]} g_{K}(s) = s^{*}$, and $\psi_{K}(s^{*}) = - \mathcal{C}(K)$, convexity of $\psi_{K}(s)$ gives
\begin{equation}
\label{eq: convexity gives tangent 0}
    - \mathcal{C}(K) = \psi_{K}(s^{*}) \geq \psi_{K}(0) + \psi_{K}'(0) s^{*} = -D(P^{K}_{U|1} \| P^{K}_{U|0}) s^{*},
\end{equation}
and
\begin{equation}
\label{eq: convexity gives tangent 1}
    - \mathcal{C}(K) = \psi_{K}(s^{*}) \geq \psi_{K}(1) + \psi_{K}'(1) (s^{*}-1) = D(P^{K}_{U|0} \| P^{K}_{U|1})(s^{*}-1).
\end{equation}
Combining \eqref{eq: D(P1 P0) D(P0 P1) leq log Gamma}, \eqref{eq: convexity gives tangent 0}, \eqref{eq: convexity gives tangent 1} derives \eqref{eq: interval of s^*} and $\mathcal{C}(K) \leq \frac{1}{2} \log \Gamma$.
\end{proof}

\begin{theorem}[Monotone convergence of Algorithm \ref{alg: Algorithm for reformulated problem}]
\label{the: convergence of algorithm_reformulated problem}
    In Algorithm \ref{alg: Algorithm for reformulated problem}, assume that, at every outer iteration $l \geq 0$, line 21 returns the exact minimizer $s^{(l+1)}$ of $B_{M}(K^{(l+1)}, \cdot)$ over $[0,1]$. Then, for every $l \geq 0$, $K^{(l)}_{U|X} \in \mathcal{K}^{(M)}_{R,+}$, $s^{(l)} \in (0,1)$, and
    \begin{equation}
    \label{eq: monotone B}
        0 < F(R) \leq B^{(l+1)} \leq B^{(l)} \leq \cdots \leq B^{(0)} < 1.
    \end{equation}
    Consequently, $\{B^{(l)}\}_{l=0}^{\infty}$ converges to some $B^{*} \in [F(R), B^{(0)}] \subset (0,1)$, and the Chernoff information $\mathcal{C}(K^{(l)}) = -\log B^{(l)}$, $l \geq 1$, is nondecreasing and converges to $-\log B^{*} \leq C(R)$.
\end{theorem}

\begin{proof}
    By lines 2 and 22, $B^{(l)} = B_{M}(K^{(l)}, s^{(l)})$ for every $l \geq 0$. For $l \geq 0$ with $K^{(l)}_{U|X} \in \mathcal{K}^{(M)}_{R,+}$ and $s^{(l)} \in (0,1)$, let $C^{s^{(l)}}_{u,x}$ be computed from $K^{(l)}_{U|X}$ in line 6, and write $G^{(l)}(K) = \sum_{u,x} C^{s^{(l)}}_{u,x} K(u|x)$. We prove by induction on $l$ that
    \begin{equation}
    \label{eq: induction claim}
        K^{(l)}_{U|X} \in \mathcal{K}^{(M)}_{R,+}, \quad s^{(l)} \in (0,1), \quad B^{(l)} \leq B^{(0)} < 1.
    \end{equation}

    \textit{Base case.} By line 1, $K^{(0)}_{U|X} \in \mathcal{K}^{(M)}_{R,+}$ and $s^{(0)} \in (0,1)$. By Young's inequality (Fact \ref{fact: Young's inequality}) with $w(u) = 1$,
    \[
    B^{(0)} \leq \sum_{u \in \mathcal{U}_{M}} \big(s^{(0)} P^{K^{(0)}}_{U|0}(u) + (1-s^{(0)}) P^{K^{(0)}}_{U|1}(u) \big) = 1,
    \]
    with equality if and only if $P^{K^{(0)}}_{U|0} \equiv P^{K^{(0)}}_{U|1}$, which is excluded by line 3. Hence $B^{(0)} < 1$.

    \textit{Induction step.} Assume \eqref{eq: induction claim} holds for $l$. In lines 7--19, $K^{(l+1)}_{U|X}$ is initialized as $K^{(l)}_{U|X}$ and replaced only by an output $\hat{K}_{U|X}$ of Algorithm \ref{alg: GBA for K_U|X} started from $K^{(l)}_{U|X}$ with $\hat{I}_{K}(U;X) \leq R$ and $G^{(l)}(\hat{K}) \leq G^{(l)}(K^{(l+1)})$. Such an output is strictly positive, by the proof of Theorem \ref{the: hatL global converge}. Hence
    \[
    K^{(l+1)}_{U|X} \in \mathcal{K}^{(M)}_{R,+}, \qquad G^{(l)}(K^{(l+1)}) \leq G^{(l)}(K^{(l)}).
    \]
    By Remark \ref{rem: concavity of B}, $G^{(l)}$ is the tangent plane of $B_{M}(\cdot, s^{(l)})$ at $K^{(l)}$, so that $B_{M}(K, s^{(l)}) \leq G^{(l)}(K)$ for every $K$, with equality at $K = K^{(l)}$. Since $s^{(l+1)}$ minimizes $B_{M}(K^{(l+1)}, \cdot)$,
    \[
    \begin{aligned}
    B^{(l+1)} = B_{M}(K^{(l+1)}, s^{(l+1)}) & \leq B_{M}(K^{(l+1)}, s^{(l)}) \leq G^{(l)}(K^{(l+1)}) \\
    & \leq G^{(l)}(K^{(l)}) = B_{M}(K^{(l)}, s^{(l)}) = B^{(l)} \leq B^{(0)} < 1.
    \end{aligned}
    \]
    If $P^{K^{(l+1)}}_{U|0} \equiv P^{K^{(l+1)}}_{U|1}$, then $B_{M}(K^{(l+1)}, \cdot) \equiv 1$, contradicting $B^{(l+1)} = \min_{s \in [0,1]} B_{M}(K^{(l+1)}, s) < 1$. Hence Lemma \ref{lem: unique S^* achieve minimum} gives $s^{(l+1)} \in (0,1)$, which completes the induction and proves the monotonicity in \eqref{eq: monotone B}.

    \textit{Lower bound.} Since $K^{(l)}_{U|X} \in \mathcal{K}^{(M)}_{R} \subseteq \mathscr{K}_{R}$,
    \[
    B^{(l)} \geq \min_{s \in [0,1]} B_{M}(K^{(l)}, s) = e^{-\mathcal{C}(K^{(l)})} \geq e^{-C(R)} = F(R),
    \]
    and $F(R) > 0$ since $C(R) \leq C_{X} < \infty$ by Lemma \ref{lem: basic shape of C(R)}.

    \textit{Convergence.} $\{B^{(l)}\}$ is nonincreasing and bounded below by $F(R)$, hence converges to some $B^{*} \in [F(R), B^{(0)}]$. For $l \geq 1$, $s^{(l)}$ minimizes $B_{M}(K^{(l)}, \cdot)$, so that $\mathcal{C}(K^{(l)}) = -\log B^{(l)}$. This sequence is nondecreasing and converges to $-\log B^{*} \leq -\log F(R) = C(R)$.
\end{proof}

\begin{remark}
\label{rem: s uniformly bounded away from 0 1}
    Lemma \ref{lem: interiority of s^*} ensures that $s^{(l)}$ lies in the interior of $(0,1)$ for each iteration $l \geq 1$ of Algorithm \ref{alg: Algorithm for reformulated problem}. Moreover, by Theorem \ref{the: convergence of algorithm_reformulated problem}, $\mathcal{C}(K^{(l)}) = -\log B^{(l)} \geq -\log B^{(0)} > 0$ for every $l \geq 1$, so that $s^{(l)}$ is uniformly bounded away from $0$ and $1$:
    \[
    s^{(l)} \in \Big[\frac{-\log B^{(0)}}{\log \Gamma}, \ 1 - \frac{-\log B^{(0)}}{\log \Gamma} \Big] \subset (0,1), \quad \forall \ l \geq 1.
    \]
\end{remark}

\section{When Does the Chernoff Bottleneck Outperform the IB?}\label{sec: CB vs IB}

The Chernoff bottleneck and the classical IB optimize different objectives over the same set of encoders. In this section, we compare their optimal encoders when $Y$ is inferred from a finite number $n$ of samples. Fix $R > 0$ and $M \geq k+1$, and let
\[
K_{\mathrm{IB}} \in \arg\max_{K \in \mathcal{K}^{(M)}_{R}} I_{K}(U;Y), \qquad K_{\mathrm{CB}} \in \arg\max_{K \in \mathcal{K}^{(M)}_{R}} \mathcal{C}(K),
\]
where both maxima exist since $\mathcal{K}^{(M)}_{R}$ is compact (Lemma~\ref{lem: K_R non empty compact convex}). By Theorem~\ref{the: sharpened cardinality}, $\mathcal{C}(K_{\mathrm{CB}}) = C(R)$. Since $I_{K}(U;Y)$ also depends on $K$ only through $P^{K}_{U|0}$ and $P^{K}_{U|1}$ (for the given $P_{Y}$), Lemma~\ref{lem: reduction to V} and the cardinality bound of the IB \cite{witsenhausen1975}, applied to $V$, show that $K_{\mathrm{IB}}$ attains the IB curve as well.

For an encoder $K$, let $U_{1}, \cdots, U_{n}$ be i.i.d.\ given $Y$ with law $P^{K}_{U|Y}$, and write $P^{K}_{U|y}(u^{n}) = \prod_{i=1}^{n} P^{K}_{U|y}(u_{i})$. Given $U^{n} = u^{n}$, the MAP test decides the $y$ with the larger $P_{Y}(y) P^{K}_{U|y}(u^{n})$, so each $u^{n}$ contributes the smaller of the two to the error probability:
\begin{equation}
\label{eq: definition Pe(n,K)}
    P_{e}(n;K) = \sum_{u^{n} \in \mathcal{U}_{M}^{n}} \min\big( P_{Y}(0) P^{K}_{U|0}(u^{n}), \ P_{Y}(1) P^{K}_{U|1}(u^{n}) \big).
\end{equation}
By Chernoff's theorem \cite{chernoff1952}, \cite[Theorem~11.9.1]{cover2006},
\begin{equation}
\label{eq: Chernoff theorem}
    \lim_{n \to \infty} -\frac{1}{n} \log P_{e}(n;K) = \mathcal{C}(K).
\end{equation}
We also consider the equivocation $H_{K}(Y|U^{n})$, i.e., the log-loss of the Bayes predictor of $Y$ from $U^{n}$. For $n = 1$, $H_{K}(Y|U) = H(Y) - I_{K}(U;Y)$, so $K_{\mathrm{IB}}$ is the best encoder for a single sample under log-loss. For large $n$, the posterior $P_{Y|U^{n}}$ is close to $0$ or $1$ with high probability, and near the endpoints the binary entropy $h(t)$ is about $t \log \frac{1}{t}$; since $\log \frac{1}{P_{e}}$ grows only linearly in $n$, the equivocation and the error probability have the same exponent.

\begin{lemma}[Equivocation and error probability]
\label{lem: equivocation vs error}
    For every encoder $K$ and $n \geq 1$, with $P_{e} = P_{e}(n;K)$,
    \begin{equation}
    \label{eq: equivocation sandwich}
        2 \log 2 \cdot P_{e} \leq H_{K}(Y|U^{n}) \leq P_{e} \big(1 - \log P_{e} \big).
    \end{equation}
    Consequently, $\lim_{n \to \infty} -\frac{1}{n} \log H_{K}(Y|U^{n}) = \mathcal{C}(K)$.
\end{lemma}

\begin{proof}
    By \eqref{eq: condition no instantly recognizable x}, $P^{K}_{U|0}(u) > 0$ if and only if $P^{K}_{U|1}(u) > 0$, so that $P_{e} > 0$. Write $P(u^{n}) = \sum_{y} P_{Y}(y) P^{K}_{U|y}(u^{n})$ and $q(u^{n}) = \min_{y} P_{Y|U^{n} = u^{n}}(y) \leq \frac{1}{2}$. Since $P_{Y}(y) P^{K}_{U|y}(u^{n}) = P(u^{n}) P_{Y|U^{n} = u^{n}}(y)$ by Bayes' rule, \eqref{eq: definition Pe(n,K)} and the definition of conditional entropy give
    \[
    P_{e} = \sum_{u^{n}} P(u^{n}) \, q(u^{n}), \qquad H_{K}(Y|U^{n}) = \sum_{u^{n}} P(u^{n}) \, h\big(q(u^{n})\big),
    \]
    where $h(t) = -t \log t - (1-t) \log (1-t)$ and we used $h(t) = h(1-t)$. Since $h$ is concave on $[0, \frac{1}{2}]$ with $h(0) = 0$ and $h(\frac{1}{2}) = \log 2$, it lies above its chord, $h(t) \geq 2 \log 2 \cdot t$, which gives the lower bound. By the concavity of $h$ and Jensen's inequality, $H_{K}(Y|U^{n}) \leq h(P_{e})$, which is Fano's inequality for binary $Y$ \cite[Theorem~2.10.1]{cover2006}. Moreover, $h(t) \leq t (1 - \log t)$, since $-(1-t) \log(1-t) \leq t$ by $\log z \leq z - 1$ with $z = \frac{1}{1-t}$. This gives the upper bound.

    By \eqref{eq: Chernoff theorem}, $-\log P_{e} = O(n)$, so that $\log(1 - \log P_{e}) = O(\log n)$. Taking logarithms in \eqref{eq: equivocation sandwich}, dividing by $n$ and using \eqref{eq: Chernoff theorem} proves the limit.
\end{proof}

\begin{theorem}[The IB versus the Chernoff bottleneck]
\label{the: IB vs CB}
    Assume $\mathcal{C}(K_{\mathrm{IB}}) < C(R)$. Then
    \begin{enumerate}[label=(\roman*)]
        \item\label{item: IB vs CB large n} $P_{e}(n;K_{\mathrm{CB}}) < P_{e}(n;K_{\mathrm{IB}})$ and $H_{K_{\mathrm{CB}}}(Y|U^{n}) < H_{K_{\mathrm{IB}}}(Y|U^{n})$ for all sufficiently large $n$;
        \item\label{item: IB vs CB crossover} if moreover $I_{K_{\mathrm{CB}}}(U;Y) < I_{K_{\mathrm{IB}}}(U;Y)$, then $H_{K_{\mathrm{IB}}}(Y|U) < H_{K_{\mathrm{CB}}}(Y|U)$, so that $n \mapsto H_{K_{\mathrm{IB}}}(Y|U^{n}) - H_{K_{\mathrm{CB}}}(Y|U^{n})$ changes sign at least once.
    \end{enumerate}
\end{theorem}

\begin{proof}
    \ref{item: IB vs CB large n} By \eqref{eq: Chernoff theorem} and Lemma~\ref{lem: equivocation vs error}, both
    \[
    \log \frac{P_{e}(n;K_{\mathrm{IB}})}{P_{e}(n;K_{\mathrm{CB}})} \quad \text{and} \quad \log \frac{H_{K_{\mathrm{IB}}}(Y|U^{n})}{H_{K_{\mathrm{CB}}}(Y|U^{n})}
    \]
    equal $n \big( C(R) - \mathcal{C}(K_{\mathrm{IB}}) \big) + o(n)$, which tends to $+\infty$.

    \ref{item: IB vs CB crossover} The inequality at $n = 1$ follows from $H_{K}(Y|U) = H(Y) - I_{K}(U;Y)$, and the sign change from \ref{item: IB vs CB large n}.
\end{proof}

The two assumptions say that the two objectives disagree at rate $R$: $K_{\mathrm{IB}}$ is not Chernoff-optimal, and $K_{\mathrm{CB}}$ is not IB-optimal. They require $R < H(V)$: for $R \geq H(V)$, the IB encoder is sufficient for $Y$ (Section~\ref{sec: global shape}), i.e., $X - U - Y$ is a Markov chain, and data processing gives $\mathcal{C}(K_{\mathrm{IB}}) \geq C_{X} = C(R)$. In terms of the learning curve $n \mapsto I_{K}(U^{n};Y) = H(Y) - H_{K}(Y|U^{n})$, the IB maximizes its first value $I_{K}(U;Y)$, while the Chernoff bottleneck maximizes the exponential rate $\mathcal{C}(K)$ at which it approaches $H(Y)$.

\begin{remark}[Error probability for small $n$]
\label{rem: Pe small n}
    Theorem~\ref{the: IB vs CB}\ref{item: IB vs CB crossover} has no counterpart for the error probability, since the IB maximizes $I_{K}(U;Y)$ rather than $1 - P_{e}(1;K)$. Indeed, using $\min(t, 1-t) = t - \max(2t-1, 0)$ with $t = P_{Y|U=u}(1)$,
    \[
    P_{e}(1;K) = P_{Y}(1) - \sum_{u:\, P_{Y|U=u}(1) > \frac{1}{2}} P^{K}_{U}(u) \big( 2 P_{Y|U=u}(1) - 1 \big).
    \]
    Assume $P_{Y}(1) \leq \frac{1}{2}$, so that $P_{Y}(1)$ is the error without observation. Then an output reduces the error only if its posterior $P_{Y|U=u}(1)$ exceeds $\frac{1}{2}$, whereas $I_{K}(U;Y) = \sum_{u} P^{K}_{U}(u) \mathrm{D}(P_{Y|U=u} \| P_{Y})$ rewards every output whose posterior differs from the prior. When $P_{Y}(1)$ is small, a posterior above $\frac{1}{2}$ requires a large likelihood ratio $P^{K}_{U|1}(u)/P^{K}_{U|0}(u)$, which can favor the Chernoff bottleneck already for small $n$; Section~\ref{sec: exp IB} gives an example.
\end{remark}

\section{Numerical Experiments}\label{sec: numerical experiments}

\subsection{The $C(R)$ Curve}\label{sec: exp C(R)}

We compute $C(R)$ for the six sources in Table~\ref{tab: sources}, all with $P_{Y}(0) = P_{Y}(1) = \frac{1}{2}$. In sources A, C, D and E, some symbols share a likelihood ratio, so that $k < |\mathcal{X}|$ and $H(V) < H(X)$, while in source B all likelihood ratios are distinct, so that $H(V) = H(X)$. Source A is symmetric, C is asymmetric, D has a weak signal, and E has a larger alphabet. Source F is the source of Example~\ref{exa: not concave} with $a = 10^{-3}$.

\begin{table}[H]
    \centering
    \footnotesize
    \caption{Sources used in Section~\ref{sec: exp C(R)}, with $P_{Y}(0) = P_{Y}(1) = \frac{1}{2}$ and $a = 10^{-3}$ (natural logarithms).}
    \label{tab: sources}
    \begin{tabular}{@{}cllcccc@{}}
    \toprule
     & $P_{X|0}$ & $P_{X|1}$ & $k$ & $H(V)$ & $H(X)$ & $C_{X}$ \\
    \midrule
    A & $(0.4,0.2,0.2,0.2)$ & $(0.2,0.2,0.2,0.4)$ & 3 & 1.089 & 1.366 & 0.0349 \\
    B & $(0.5,0.3,0.2)$ & $(0.2,0.3,0.5)$ & 3 & 1.096 & 1.096 & 0.0699 \\
    C & $(0.5,0.2,0.2,0.1)$ & $(0.1,0.3,0.3,0.3)$ & 3 & 1.030 & 1.376 & 0.1208 \\
    D & $(0.3,0.25,0.25,0.2)$ & $(0.25,0.25,0.25,0.25)$ & 3 & 1.037 & 1.384 & 0.0025 \\
    E & $(0.2,0.15,0.15,0.1,0.1,0.1,0.1,0.1)$ & $(0.05,0.075,0.075,0.1,0.1,0.2,0.2,0.2)$ & 4 & 1.277 & 2.066 & 0.0658 \\
    F & $(1-a,a)$ & $(a,1-a)$ & 2 & 0.693 & 0.693 & 2.7612 \\
    \bottomrule
    \end{tabular}
\end{table}

For each source we set $M = k+1$, which suffices by Theorem~\ref{the: sharpened cardinality}, and run Algorithm~\ref{alg: Algorithm for reformulated problem} at $50$ evenly spaced rates in $(0, H(V)]$ and, when $H(V) < H(X)$, at $10$ rates in $(H(V), H(X)]$. Since maximizing $\mathcal{C}(K)$ amounts to minimizing a concave function (Remark~\ref{rem: concavity of B}), the algorithm may stop at a local optimum, so at each rate we keep the best of $k+3$ initializations: two random encoders; the time-sharing encoder of Lemma~\ref{lem: time sharing} between $K^{V}$ and a constant encoder; and, for each $j = 1, \cdots, k$, the encoder that outputs $j$ with probability $\theta$ when $X \in \mathcal{X}_{j}$ and a constant symbol otherwise. In the last two, $\theta$ is the largest value with $I_{K}(U;X) \leq R$, and the encoder is mixed with weight $10^{-9}$ with a random feasible encoder, so that it lies in $\mathcal{K}^{(M)}_{R,+}$. Every plotted value is attained by an explicit encoder and is therefore a lower bound on $C(R)$.

\begin{figure}[t]
    \centering
    \includegraphics[width=\linewidth]{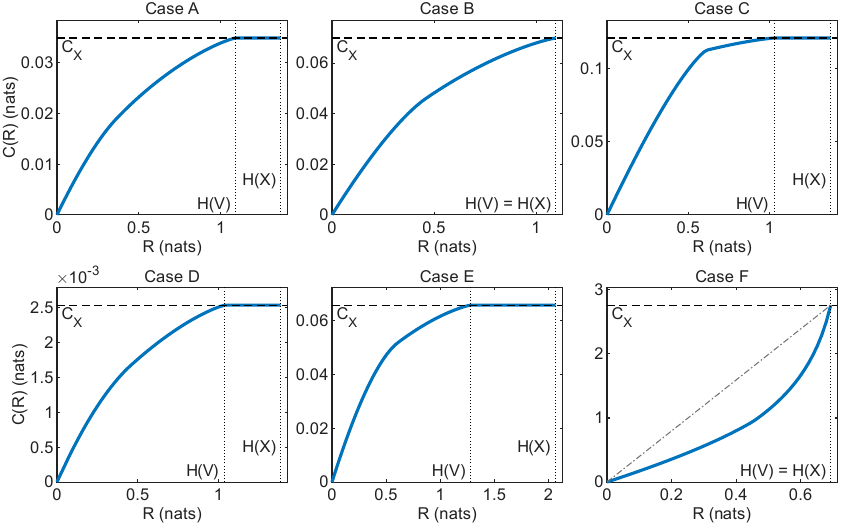}
    \caption{$C(R)$ for the sources of Table~\ref{tab: sources}. Dashed: $C_{X}$. Dotted: $H(V)$ and $H(X)$. Dash-dotted (F): the chord from $(0,0)$ to $(H(V), C_{X})$.}
    \label{fig: C(R) curves}
\end{figure}

Figure~\ref{fig: C(R) curves} shows the results. In every case, $C(R)$ is strictly increasing on $[0, H(V)]$ and equals $C_{X}$ on $[H(V), H(X)]$, in agreement with Theorem~\ref{the: saturation and strict monotonicity}: rate beyond $H(V)$ only distinguishes symbols with the same likelihood ratio, which does not help the test. For source B, $H(V) = H(X)$ and the curve has no flat part. Source F illustrates Example~\ref{exa: not concave}: its whole curve lies below the chord from $(0,0)$ to $(H(V), C_{X})$. Since B is concave while both B and F have $H(V) = H(X)$, the non-concavity is not a consequence of $H(V) = H(X)$.

\subsection{Structure of the Optimal Encoder}\label{sec: exp structure}

We examine the optimal encoders of source C in Table~\ref{tab: sources}, whose classes $\mathcal{X}_{1} = \{1\}$, $\mathcal{X}_{2} = \{2,3\}$, $\mathcal{X}_{3} = \{4\}$ have likelihood ratios $\rho_{1} = 0.2$, $\rho_{2} = 1.5$, $\rho_{3} = 3$, and whose curve in Figure~\ref{fig: C(R) curves} has a kink near $R = 0.61$. We take three rates: $R = 0.30$; $R = h(P_{V}(1)) = h(0.3) \approx 0.61$, where $h$ is the binary entropy function, which is the rate of the deterministic encoder that separates $\mathcal{X}_{1}$ from the other classes; and $R = 0.80$. At each rate we compute the best encoder $K^{*}$ as in Section~\ref{sec: exp C(R)} and map it, as in Lemma~\ref{lem: reduction to V}(i), to the encoder of $V$
\[
\bar{K}^{*}(u|j) = \sum_{x \in \mathcal{X}_{j}} \frac{P_{X}(x)}{P_{V}(j)} K^{*}(u|x),
\]
which has the same $P_{U|0}$ and $P_{U|1}$ and no larger rate. Outputs with $P_{U}(u) \leq 10^{-4}$ are dropped, outputs with equal likelihood ratio are merged, and the remaining outputs are sorted by $P_{U|1}(u)/P_{U|0}(u)$.

\begin{figure}[t]
    \centering
    \includegraphics[width=\linewidth]{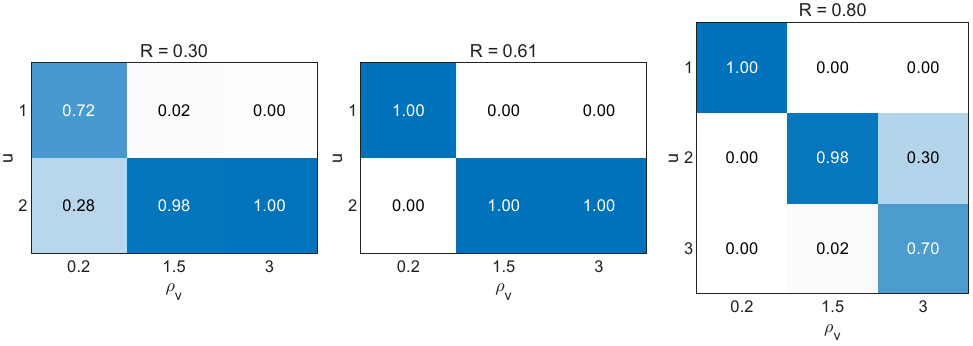}
    \caption{Optimal encoders $\bar{K}^{*}(u|v)$ of source C at three rates. Columns: classes $v$, labeled by $\rho_{v}$. Rows: used outputs $u$, sorted by $P_{U|1}(u)/P_{U|0}(u)$.}
    \label{fig: structure}
\end{figure}

Figure~\ref{fig: structure} shows three features. First, the optimal encoders are randomized: at $R = 0.30$, class $1$ is sent to output $1$ with probability $0.72$, and at $R = 0.80$, class $3$ is split between outputs $2$ and $3$ with probabilities $0.30$ and $0.70$. This differs from quantizer design under an alphabet constraint, where deterministic likelihood-ratio quantizers are optimal \cite{tsitsiklis1993}. Second, the encoders are monotone in the likelihood ratio: each output collects a contiguous range of classes in the order of $\rho_{v}$, so that $\bar{K}^{*}$ is a randomized likelihood-ratio quantizer. The same randomized, monotone structure appears for source E. Third, at $R = h(0.3)$ the optimal encoder is the deterministic quantizer $\{\mathcal{X}_{1}\}, \{\mathcal{X}_{2} \cup \mathcal{X}_{3}\}$, and this rate is exactly where the curve of source C has its kink. Below it, the encoder randomizes the separation of $\mathcal{X}_{1}$; above it, $\mathcal{X}_{1}$ stays separated and the encoder randomizes the separation of $\mathcal{X}_{2}$ and $\mathcal{X}_{3}$ with a third output.

\subsection{Real Dataset Practice}\label{sec: exp real}

We apply the Chernoff bottleneck to topic detection, where a document is observed word by word and every word is compressed separately before the decision, as in decentralized detection \cite{tsitsiklis1993dd, chamberland2003}. We use the 20 Newsgroups corpus \cite{lang1995} in its date-sorted train/test split, with the topics \texttt{rec.sport.hockey} ($Y = 0$) and \texttt{sci.med} ($Y = 1$), which have $598$ and $594$ training documents and $399$ and $393$ test documents. The alphabet $\mathcal{X}$ consists of the $1000$ most frequent words in the training documents of the two topics, and all other words are discarded. From the training word counts $N_{y}(x)$ we estimate
\[
P_{X|y}(x) = \frac{N_{y}(x) + 1}{\sum_{x'} N_{y}(x') + |\mathcal{X}|}, \quad y = 0, 1,
\]
which satisfies \eqref{eq: condition no instantly recognizable x}, and set $P_{Y}(0) = P_{Y}(1) = \frac{1}{2}$. This gives $k = 826$, $H(X) = 5.80$ and $C_{X} = 0.158$ nats per word.

We compute encoders with $M = 16$ outputs at $R \in \{0.25, 0.5, 1, 2\}$ nats per word by Algorithm~\ref{alg: Algorithm for reformulated problem}, keeping at each rate the best of four initializations: two random encoders, the encoder of the previous rate, and a likelihood-ratio quantizer that divides the words, in the order of $\rho(x)$, into $M - 1$ groups of equal $P_{X}$-mass and is time-shared with a constant symbol so that $I_{K}(U;X) \leq R$. Since $M < k+1$, the resulting values of $\mathcal{C}(K)$, reported as $C(R)$ in Table~\ref{tab: newsgroups}, are lower bounds on $C(R)$. Each word is encoded independently by the same $K$, and the decision on $U^{n} = (U_{1}, \cdots, U_{n})$ is the MAP test
\[
\hat{Y} = 1 \iff \sum_{i=1}^{n} \log \frac{P^{K}_{U|1}(U_{i})}{P^{K}_{U|0}(U_{i})} > 0,
\]
where $P^{K}_{U|y}$ is computed from the training estimates. Without compression ($U_{i} = X_{i}$), this test is the naive Bayes classifier. We measure the error on the test data in two ways, in both cases averaged over the two topics: (a) the $n$ words are drawn i.i.d.\ from the word frequencies of each topic in the test documents ($2 \times 10^{5}$ repetitions per $n$ and topic), so that the model of Section~\ref{sec: Setup} holds up to the mismatch between training and test data; (b) the $n$ words are drawn with replacement from a single test document ($200$ repetitions per document), which is the actual task.

\begin{figure}[t]
    \centering
    \includegraphics[width=\linewidth]{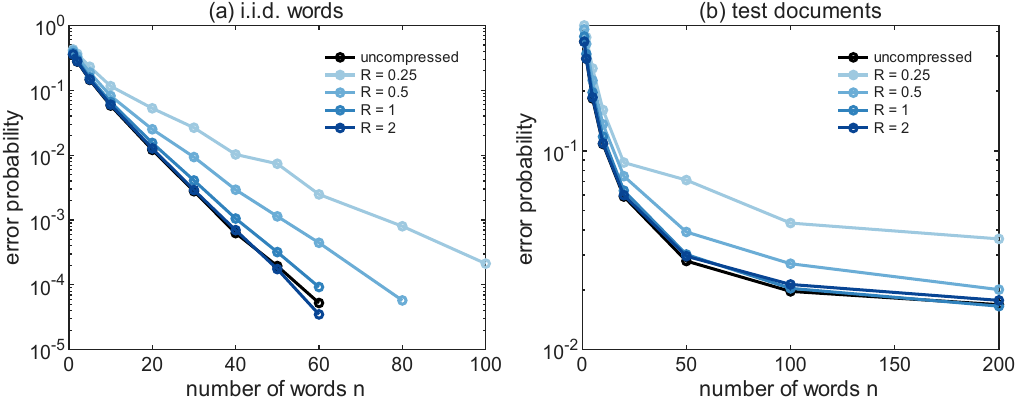}
    \caption{Test error of topic detection (hockey vs.\ medicine) from $n$ words, each compressed to $R$ nats by the Chernoff bottleneck, compared with the uncompressed naive Bayes classifier. (a) Words drawn i.i.d.\ from each topic. (b) Words drawn from each test document.}
    \label{fig: newsgroups}
\end{figure}

\begin{table}[t]
    \centering
    \caption{Rate, Chernoff information, and the fitted decay rate of the error in Figure~\ref{fig: newsgroups}(a), i.e., the slope of $-\log P_{e}$ in $n$ over $10^{-4} \leq P_{e} \leq 10^{-1}$ (nats).}
    \label{tab: newsgroups}
    \setlength{\tabcolsep}{5pt}
    \renewcommand{\arraystretch}{1.1}
    \begin{tabular}{@{}ccccc@{}}
    \toprule
    $R$ & $R / H(X)$ & $C(R)$ & $C(R) / C_{X}$ & fitted slope \\
    \midrule
    $0.25$ & $0.04$ & $0.073$ & $0.47$ & $0.071$ \\
    $0.5$ & $0.09$ & $0.113$ & $0.71$ & $0.107$ \\
    $1$ & $0.17$ & $0.142$ & $0.90$ & $0.128$ \\
    $2$ & $0.34$ & $0.154$ & $0.98$ & $0.141$ \\
    uncompressed & $1$ & $0.158$ & $1$ & $0.151$ \\
    \bottomrule
    \end{tabular}
\end{table}

The computed encoders use $3$, $8$, $12$ and $12$ of the $16$ outputs, and all of them are randomized likelihood-ratio quantizers, as in Section~\ref{sec: exp structure}: each output collects words of similar $\rho(x)$, and the assignment is monotone in $\rho(x)$. Table~\ref{tab: newsgroups} shows how much of the exponent is kept: at $R = 1$ nat per word, which is $17\%$ of $H(X)$, the encoder keeps $90\%$ of $C_{X}$, and even at $R = 0.25$, which is $4\%$ of $H(X)$, it keeps $47\%$.

Figure~\ref{fig: newsgroups}(a) confirms that $C(R)$ predicts the performance on real word statistics: the error decays exponentially in $n$, with fitted rates within about $10\%$ of $C(R)$, the small shortfall being due to the polynomial factor in $n$ and the mismatch between training and test data. Consequently, compression to $R$ nats per word multiplies the number of words needed for a given error by about $C_{X}/C(R)$, i.e., by $1.11$ at $R = 1$ and $1.40$ at $R = 0.5$. Figure~\ref{fig: newsgroups}(b) shows the actual task. The curves are ordered as $C(R)$; at $R = 1$ the error is close to that of the uncompressed classifier for all $n$, and at $R = 2$ the two are indistinguishable. The error levels off near $2\%$ rather than decaying as $e^{-nC}$, because the words of a document follow the document's own distribution rather than being i.i.d.\ given the topic; the uncompressed classifier has the same floor, so it is not caused by compression. Finally, the benefit of the mutual information constraint is largest at low rates: at $R = 0.25$, the best deterministic likelihood-ratio quantizer with at most five cells and entropy at most $R$ has Chernoff information $0.059$, against $0.073$ for the Chernoff bottleneck, whereas at $R \geq 1$ both are within a few percent of $C_{X}$.

\subsection{Comparison with the IB}\label{sec: exp IB}

We compare the encoders $K_{\mathrm{IB}}$ and $K_{\mathrm{CB}}$ of Section~\ref{sec: CB vs IB}. When all likelihood ratios are moderate, the two nearly coincide, so we use the source
\[
P_{X|0} = (0.26, 0.37, 0.365, 0.005), \qquad P_{X|1} = (0.25, 0.02, 0.005, 0.725),
\]
with $P_{Y}(0) = P_{Y}(1) = \frac{1}{2}$, whose likelihood ratios $\rho(x) = 0.96, 0.054, 0.014, 145$ span a wide range, so that $k = 4$, $H(V) = H(X) = 1.347$ and $C_{X} = 0.820$. We take $M = k+1 = 5$ and the low rates $R = 0.0168, 0.0337, 0.0674, 0.1347$.

The encoders $K_{\mathrm{CB}}$ are computed by Algorithm~\ref{alg: Algorithm for reformulated problem}, and $K_{\mathrm{IB}}$ by the same scheme with $B_{M}(\cdot, s)$ replaced by $I_{K}(U;Y)$: since $I_{K}(U;Y)$ is convex in $K$, maximizing its tangent plane at the current encoder is the problem \eqref{eq: sum u,x C K} with $C^{s}_{u,x}$ replaced by $P_{X}(x) \mathrm{D}(P_{Y|X=x} \| P_{Y|U=u})$, which Algorithm~\ref{alg: GBA for K_U|X} solves with the same bisection and safeguard. Unlike the iterative IB algorithm \cite{tishby2000}, the bisection meets the rate constraint exactly. For both problems, we keep the best of several initializations, including the best encoder of the other problem. For $n = 1, \cdots, 500$, $P_{e}(n;K)$ and $H_{K}(Y|U^{n})$ are computed without Monte Carlo, from the law of the sum of the log-likelihood ratios of $U_{1}, \cdots, U_{n}$, obtained by $n$-fold convolution via the FFT.

\begin{figure}[t]
    \centering
    \includegraphics[width=\linewidth]{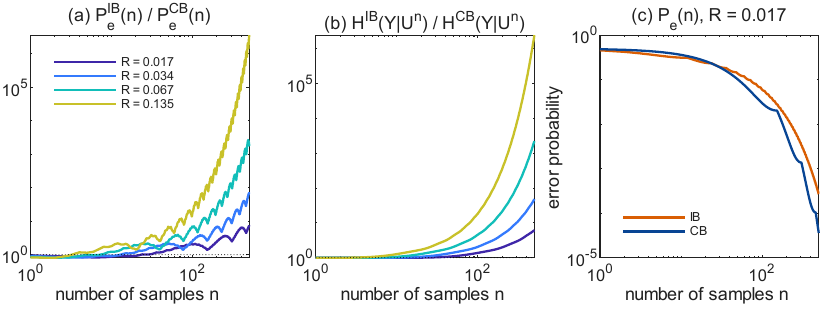}
    \caption{IB versus Chernoff bottleneck (CB) encoders. (a) Ratio of error probabilities $P_{e}(n;K_{\mathrm{IB}}) / P_{e}(n;K_{\mathrm{CB}})$. (b) Ratio of equivocations $H_{K_{\mathrm{IB}}}(Y|U^{n}) / H_{K_{\mathrm{CB}}}(Y|U^{n})$. Values above $1$ (dotted) mean that the Chernoff bottleneck is better. (c) Error probabilities at $R = 0.017$.}
    \label{fig: IB vs CB}
\end{figure}

\begin{table}[t]
    \centering
    \caption{IB versus Chernoff bottleneck. $n_{c}$: the IB has the smaller $H_{K}(Y|U^{n})$ or $P_{e}(n;K)$ exactly for $n < n_{c}$. $P_{e}(n_{c})$: error of $K_{\mathrm{CB}}$ at the crossover. $n_{\epsilon}$: the smallest $n$ with $P_{e}(n;K) \leq \epsilon$, for $K_{\mathrm{IB}}$ / $K_{\mathrm{CB}}$.}
    \label{tab: IB vs CB}
    \setlength{\tabcolsep}{5pt}
    \renewcommand{\arraystretch}{1.1}
    \begin{tabular}{@{}cccccccc@{}}
    \toprule
    & & & \multicolumn{2}{c}{$n_{c}$} & & & \\
    \cmidrule(lr){4-5}
    $R$ & $C(R)$ & $\mathcal{C}(K_{\mathrm{IB}})$ & $H$ & $P_{e}$ & $P_{e}(n_{c})$ & $n_{10^{-2}}$ & $n_{10^{-4}}$ \\
    \midrule
    $0.0168$ & $0.0157$ & $0.0120$ & $18$ & $24$ & $0.23$ & $228 \,/\, 180$ & $> 500 \,/\, 460$ \\
    $0.0337$ & $0.0318$ & $0.0240$ & $8$ & $14$ & $0.20$ & $112 \,/\, 89$ & $288 \,/\, 226$ \\
    $0.0674$ & $0.0638$ & $0.0483$ & $5$ & $8$ & $0.18$ & $57 \,/\, 45$ & $143 \,/\, 107$ \\
    $0.1347$ & $0.1281$ & $0.0986$ & $3$ & $4$ & $0.17$ & $28 \,/\, 23$ & $71 \,/\, 53$ \\
    \bottomrule
    \end{tabular}
\end{table}

At every rate, the IB has the larger $I_{K}(U;Y)$ and the Chernoff bottleneck an exponent larger by $23$--$25\%$, so Theorem~\ref{the: IB vs CB} applies. As predicted, the ratio of equivocations in Figure~\ref{fig: IB vs CB}(b) starts below $1$ and crosses $1$ exactly once. The error probabilities behave in the same way, with crossovers $n_{c} = 24, 14, 8, 4$. The crossover occurs while the error is still large, at $n_{c} C(R) \approx 0.4$--$0.5$ and $P_{e} \approx 0.2$, and $n_{c}$ grows like $1/R$. For smaller errors, the Chernoff bottleneck needs $18$--$25\%$ fewer samples (Table~\ref{tab: IB vs CB}), close to the asymptotic saving $1 - \mathcal{C}(K_{\mathrm{IB}})/C(R) \approx 24\%$.

The encoders explain this. Both are essentially binary, with one output, the alarm, carrying the evidence for $Y = 1$. The Chernoff bottleneck sends the alarm only for $x = 4$, rarely ($P^{K}_{U}(u) = 0.016$ to $0.12$) but with likelihood ratio $\rho(4) = 145$; the IB sends it also for the other symbols, more often ($0.09$ to $0.29$), but with likelihood ratio only $2.8$ to $4.6$. A frequent but weak alarm is better for one sample, a rare but decisive one for many. Since $P_{e}(n;K)$ depends essentially only on the number of alarms, it decreases in steps, which causes the oscillations in Figure~\ref{fig: IB vs CB}(a).

Finally, as Remark~\ref{rem: Pe small n} indicates, the small-$n$ advantage of the IB in error probability depends on the prior. For the same source with $P_{Y}(0) = 0.9$ and $R = H(V)/10 = 0.128$, the IB has the larger $I_{K}(U;Y)$, yet the Chernoff bottleneck has the smaller error for $n = 1, \cdots, 5$, e.g., $0.0229$ against $0.0341$ at $n = 4$. Both alarms have posterior $P_{Y|U=u}(1)$ above $\frac{1}{2}$, and the purer one of the Chernoff bottleneck ($0.94$ against $0.88$) reduces $P_{e}(1;K)$ more.

\section{Conclusion}\label{sec: conclusion}

We studied the MI constrained Chernoff bottleneck, which replaces the relevance $I(U;Y)$ of the classical IB by the Chernoff information of the downstream binary test. The curve $C(R)$ is governed by the variable $V$ that merges the symbols of $X$ with equal likelihood ratio: $C(R)$ increases strictly up to $R = H(V)$, where it reaches the uncompressed exponent $C_{X}$, and stays constant beyond, indicating the additional rates beyond $H(V)$ is wasted. (Theorem~\ref{the: saturation and strict monotonicity}). For $R \leq H(V)$, every optimal encoder uses the full rate budget (Corollary~\ref{cor: rate constraint active}). Unlike the classical IB curve, which is concave by time-sharing, the $C(R)$ curve need not be concave (Example~\ref{exa: not concave}). An optimal encoder exists with at most $|\mathcal{X}| + 1$ outputs (Theorem~\ref{the: supp U leq X+1}) and, by the reduction to $V$, with at most $k + 1$ outputs, where $k$ is the number of distinct likelihood ratios (Theorem~\ref{the: sharpened cardinality}).

Since the Chernoff information is not an expectation of a per-letter quantity, the classical Blahut--Arimoto algorithm does not apply directly. We therefore proposed Algorithm~\ref{alg: Algorithm for reformulated problem}, which alternates between the encoder $K$ and the parameter $s$. For fixed $s$, the encoder minimizes the tangent plane of the concave function $B_{M}(\cdot, s)$ at the current encoder under the MI constraint (Remark~\ref{rem: concavity of B}), which is a rate--distortion problem solved by a generalized Blahut--Arimoto algorithm (Algorithm~\ref{alg: GBA for K_U|X}) with a bisection over the trade-off parameter $\lambda$. Algorithm~\ref{alg: GBA for K_U|X} converges to a global minimizer of the Lagrangian at rate $O(1/t)$ with a stopping certificate (Theorem~\ref{the: hatL global converge}, Corollary~\ref{cor: GBA stopping certificate}), and the bisection has an explicit bracket and certificate (Lemma~\ref{lem: bisection search of lambda}). For fixed encoder $K$, the parameter $s$ is the unique root of a nonlinear equation (Lemma~\ref{lem: unique S^* achieve minimum}). The iterates of Algorithm~\ref{alg: Algorithm for reformulated problem} remain feasible, $s$ stays uniformly bounded away from $0$ and $1$, and the Chernoff information is nondecreasing and convergent (Lemma~\ref{lem: interiority of s^*}, Theorem~\ref{the: convergence of algorithm_reformulated problem}, Remark~\ref{rem: s uniformly bounded away from 0 1}). Since the problem amounts to minimizing a concave function, the limit need not be globally optimal, which we mitigate by structured initializations.

We also compared the Chernoff bottleneck with the classical IB when $Y$ is inferred from $n$ samples. The equivocation $H(Y|U^{n})$ decays at the same exponential rate as the error probability, namely the Chernoff information (Lemma~\ref{lem: equivocation vs error}). Hence, when the two objectives disagree, the IB encoder, which is the best for a single sample under log-loss, is overtaken by the Chernoff bottleneck as $n$ grows, in both equivocation and error probability (Theorem~\ref{the: IB vs CB}). For the error probability, the IB need not be better even for small $n$, since an output reduces the error only if it moves the posterior across $\frac{1}{2}$ (Remark~\ref{rem: Pe small n}).

The numerical experiments confirm the theory. Figure~\ref{fig: C(R) curves} shows the computed $C(R)$ curves of six synthetic sources: each curve increases strictly up to $H(V)$ and stays at $C_{X}$ beyond; for source B, whose likelihood ratios are all distinct, $H(V) = H(X)$ and the curve has no flat part; and the curve of source F, which is the source of Example~\ref{exa: not concave}, lies entirely below its chord (Section~\ref{sec: exp C(R)}). Figure~\ref{fig: structure} shows the computed encoders of source C at three rates. They are randomized likelihood-ratio quantizers, each output collecting a contiguous range of likelihood ratios, in contrast to quantizer design under an alphabet constraint, where deterministic likelihood-ratio quantizers are optimal. At the kink of $C(R)$, the encoder is exactly a deterministic quantizer (Section~\ref{sec: exp structure}). In a real-data application, topic detection on the 20 Newsgroups corpus, compressing each word to $1$ nat, only $17\%$ of its entropy, keeps $90\%$ of the error exponent (Table~\ref{tab: newsgroups}). When the words are drawn i.i.d.\ from the test word frequencies, the error decays exponentially at a rate within about $10\%$ of $C(R)$ (Figure~\ref{fig: newsgroups}(a)). On the actual test documents, the errors are ordered as $C(R)$, and at $1$ nat the classifier is nearly as accurate as the uncompressed one (Figure~\ref{fig: newsgroups}(b)). The advantage over deterministic likelihood-ratio quantization is largest at low rates (Section~\ref{sec: exp real}). Compared with the IB on a source with widely spread likelihood ratios, the IB has the smaller error only while the error is about $0.2$ or larger, and below that the Chernoff bottleneck needs $18$--$25\%$ fewer samples (Section~\ref{sec: exp IB}).

More broadly, the Chernoff bottleneck judges a representation by how fast the error of the downstream decision decays, rather than by an average measure of relevance. We believe this makes it a principled criterion for learning compact representations for inference, and hope that the theory, algorithm, and experiments presented here serve as a foundation for its further study.

\newpage

\appendix

\section{Several Facts Used in the Proof}

\begin{fact}[Young's inequality]
\label{fact: Young's inequality}
For any $x,y > 0$, and $p,q > 1$ with $\frac{1}{p} + \frac{1}{q} = 1$, we have
\[
xy \leq \frac{x^{p}}{p} + \frac{y^{q}}{q},
\]
equality holds when $x^p = y^q$.
\end{fact}

\begin{fact}[H\"older's inequality]
\label{fact: H\"older's inequality}
For nonnegative $a_{1}, \cdots, a_{n}$, $b_{1}, \cdots, b_{n}$, and $p, q > 1$ with $\frac{1}{p} + \frac{1}{q} = 1$,
\[
\sum_{i=1}^{n} a_{i} b_{i} \leq \Big(\sum_{i=1}^{n} a_{i}^{p}\Big)^{1/p} \Big(\sum_{i=1}^{n} b_{i}^{q}\Big)^{1/q},
\]
with equality if and only if there exist $\alpha, \beta \geq 0$, not both zero, such that $\alpha a_{i}^{p} = \beta b_{i}^{q}$ for all $i$.
\end{fact}

\begin{fact}[Weighted average inequality]
\label{fact: Weighted average inequality}
For finite $k$ and $a_{i}, b_{i} > 0$, $i = 1, \cdots, k$, we have
\[
\min_{1 \leq i \leq k} \frac{a_{i}}{b_{i}} \leq \frac{\sum_{i=1}^{k} a_{i}}{\sum_{i=1}^{k} b_{i}} \leq \max_{1 \leq i \leq k} \frac{a_{i}}{b_{i}}.
\]
\end{fact}

\begin{fact}[Heine-Borel Theorem]
\label{fact: Heine-Borel Theorem}
    A subset of $\mathbb{R}^{n}$ is compact if and only if it is closed and bounded. \cite[Theorem~2.41]{rudin1976}
\end{fact}

\begin{fact}[Joint convexity of relative entropy]
\label{fact: joint convexity of relative entropy}
    For probability distributions $P_{i}, Q_{i}$, $i = 1, \cdots, n$, on a finite set and weights $\pi_{i} \geq 0$ with $\sum_{i} \pi_{i} = 1$,
    \[
    \mathrm{D}\Big(\sum_{i=1}^{n} \pi_{i} P_{i} \,\Big\|\, \sum_{i=1}^{n} \pi_{i} Q_{i}\Big) \leq \sum_{i=1}^{n} \pi_{i} \mathrm{D}(P_{i} \| Q_{i}).
    \]
    \cite[Theorem~2.7.2]{cover2006}
\end{fact}

\begin{fact}[Pinsker's inequality]
\label{fact: Pinsker's inequality}
    For probability distributions $P, Q$ on a finite set, with natural logarithms,
    \[
    \mathrm{D}(P \| Q) \geq \frac{1}{2} \|P - Q\|_{1}^{2}.
    \]
    \cite[Lemma~11.6.1]{cover2006}
\end{fact}

\section{Proofs}

\subsection{Proof of Lemma~\ref{lem: KV preserves Bs}}\label{proof: KV preserves Bs}

$V = \nu(X)$ gives $H(V|X) = 0$. Thus, $I_{K^{V}}(U;X) = H(V) - H(V|X) = H(V)$.

Fix $j \in \{1, \cdots, k\}$ and $s \in [0,1]$. On the class $\mathcal{X}_{j}$ we have $P_{X|1}(x) = \rho_{j} P_{X|0}(x)$, hence
\[
\sum_{x \in \mathcal{X}_{j}} P_{X|0}(x)^{s} P_{X|1}(x)^{1-s} = \sum_{x \in \mathcal{X}_{j}} P_{X|0}(x)^{s} \big(\rho_{j} P_{X|0}(x) \big)^{1-s} = \rho_{j}^{1-s} P_{X|0}(\mathcal{X}_{j}).
\]
And since $P^{K^{V}}_{U|y}(j) = P_{X|y}(\mathcal{X}_{j})$ for $y \in \{0,1\}$, and $P_{X|1}(\mathcal{X}_{j}) = \rho_{j} P_{X|0}(\mathcal{X}_{j})$, we derive
\[
P^{K^{V}}_{U|0}(j)^{s} P^{K^{V}}_{U|1}(j)^{1-s} = P_{X|0}(\mathcal{X}_{j})^{s} \big(\rho_{j} P_{X|0}(\mathcal{X}_{j}) \big)^{1-s} = \rho_{j}^{1-s} P_{X|0}(\mathcal{X}_{j}) = \sum_{x \in \mathcal{X}_{j}} P_{X|0}(x)^{s} P_{X|1}(x)^{1-s}.
\]
Summing over $j$ and using $\mathcal{X} = \bigsqcup_{j} \mathcal{X}_{j}$ yields $B_{k}(K^{V},s) = B_{|\mathcal{X}|}(K^{X},s)$. Taking the minimum over $s \in [0,1]$ gives $C(P^{K^{V}}_{U|0}, P^{K^{V}}_{U|1}) = C_{X}$.

\subsection{Proof of Lemma~\ref{lem: holder BsK geq BsX}}\label{proof: holder BsK geq BsX}

For $s \in \{0,1\}$ both sides of \eqref{eq: BsK geq BsX} equal $1$, so let $s \in (0,1)$. Applying H\"older's inequality (Fact \ref{fact: H\"older's inequality}) with exponents $\frac{1}{s}$ and $\frac{1}{1-s}$ to $\big(K(u|x) P_{X|0}(x) \big)^{s}$ and $\big(K(u|x) P_{X|1}(x) \big)^{1-s}$ for all $x \in \mathcal{X}$ gives, for each $u \in \mathcal{U}_{M}$,
    \begin{equation}
    \label{eq: holder per symbol}
    \begin{aligned}
        \sum_{x} K(u|x) P_{X|0}(x)^{s} P_{X|1}(x)^{1-s}
        & \leq \Big(\sum_{x} K(u|x) P_{X|0}(x) \Big)^{s} \Big(\sum_{x} K(u|x) P_{X|1}(x) \Big)^{1-s} \\
        &= P^{K}_{U|0}(u)^{s} P^{K}_{U|1}(u)^{1-s}.
    \end{aligned}
    \end{equation}
    Summing \eqref{eq: holder per symbol} over $u \in \mathcal{U}_{M}$ and using $\sum_{u \in \mathcal{U}_{M}} K(u|x) = 1$ derives \eqref{eq: BsK geq BsX}.

    Suppose that $B_{M}(K,s) = B_{|\mathcal{X}|}(K^{X},s)$ for some $s \in (0,1)$. Then \eqref{eq: holder per symbol} holds with equality for every $u \in \mathcal{U}_{M}$. By \eqref{eq: condition no instantly recognizable x}, for any $u \in \mathcal{U}_{M}$ with $P^{K}_{U}(u) > 0$, the equality condition of Fact \ref{fact: H\"older's inequality} indicates there exists $c_{u} > 0$ such that
    \[
    K(u|x) P_{X|1}(x) = c_{u} K(u|x) P_{X|0}(x), \quad \forall x \in \mathcal{X}.
    \]
    This means $\rho(x) = c_{u}$ for all $x$ with $K(u|x) > 0$, which implies $U$ determines $V = \nu(X)$, i.e., $V = g(U)$ almost surely for some function $g$. Hence,
    \[
    I_{K}(U;X) \geq I\big(g(U); X \big) = I(V;X) = H(V) - H(V|X) = H(V).
    \]
    Taking the contrapositive gives the conclusion.

\subsection{Proof of Lemma~\ref{lem: time sharing}}\label{proof: time sharing}

Since $Q = \mathbf{1}\{U \leq M_1\}$ is a function of $U$, and $Q$ is independent of $X$, we have
\begin{align*}
    I_{K}(U;X) & = I(Q;X) + I(U;X|Q) \\ & = P(Q = 1)I(U;X|Q = 1) + P(Q = 0)I(U;X|Q = 0) \\ & = \theta I_{K_1}(U;X) + (1-\theta) I_{K_2}(U;X).
\end{align*}
For $y \in \{0,1\}$, $P^{K}_{U|y}(u) = \theta P^{K_1}_{U|y}(u)$ for $u \leq M_1$ and $P^{K}_{U|y}(u) = (1-\theta) P^{K_2}_{U|y}(u - M_1)$ for $u > M_1$. Using $(c\alpha)^{s} (c\beta)^{1-s} = c \, \alpha^{s} \beta^{1-s}$ for $c, \alpha, \beta \geq 0$, we have
\[
\begin{aligned}
B_{M_1+M_2}(K,s)
&= \sum_{u=1}^{M_1} \big(\theta P^{K_1}_{U|0}(u)\big)^{s} \big(\theta P^{K_1}_{U|1}(u)\big)^{1-s} + \sum_{u=M_1+1}^{M_1+M_2} \big((1-\theta) P^{K_2}_{U|0}(u-M_1)\big)^{s} \big((1-\theta) P^{K_2}_{U|1}(u-M_1)\big)^{1-s} \\
&= \theta \sum_{u=1}^{M_1} P^{K_1}_{U|0}(u)^{s} P^{K_1}_{U|1}(u)^{1-s} + (1-\theta) \sum_{u=1}^{M_2} P^{K_2}_{U|0}(u)^{s} P^{K_2}_{U|1}(u)^{1-s} \\
&= \theta B_{M_1}(K_1,s) + (1-\theta) B_{M_2}(K_2,s).
\end{aligned}
\]

\subsection{Proof of Example~\ref{exa: not concave}}\label{proof: counterexample}

Here $V = X$, so $H(V) = \log 2$, and by Theorem~\ref{the: saturation and strict monotonicity} and symmetry, $C(\log 2) = C_{X} = -\log\big(2\sqrt{a(1-a)}\big) \geq 2.76$. Since $C(0) = 0$, concavity would give $C(R_{0}) \geq \frac{1}{10} C_{X} \geq 0.276$ at $R_{0} = \frac{1}{10} \log 2$. We show $C(R_{0}) < 0.26$.

Let $K$ be an encoder with $I_{K}(U;X) \leq R_{0}$, write $K_{x} = K(\cdot|x)$ and $\beta(P,Q) = \sum_{u} \sqrt{P(u) Q(u)}$.

(i) $\psi(s) = \log B_{M}(K,s)$ is convex with $\psi(0) = \psi(1) = 0$. For $s \leq \frac{1}{2}$, convexity on $[s, 1]$ gives $\psi(\frac{1}{2}) \leq \frac{\psi(s)}{2(1-s)}$, so $\psi(s) \geq 2\psi(\frac{1}{2})$, and symmetrically for $s \geq \frac{1}{2}$. Hence $\mathcal{C}(K) \leq -2 \log \beta(P^{K}_{U|0}, P^{K}_{U|1})$.

(ii) $P^{K}_{U|0} = (1-a) K_{1} + a K_{2}$, $P^{K}_{U|1} = a K_{1} + (1-a) K_{2}$, and $\beta$ is jointly concave, so $\beta(P^{K}_{U|0}, P^{K}_{U|1}) \geq \beta(K_{1}, K_{2})$.

(iii) With $m = \frac{K_{1} + K_{2}}{2}$, $t_{u} = \frac{K_{1}(u)}{2 m(u)}$ and the binary entropy $h$, we have $I_{K}(U;X) = \sum_{u} m(u) \big(\log 2 - h(t_{u})\big)$ and $\beta(K_{1}, K_{2}) = \sum_{u} 2 m(u) \sqrt{t_{u}(1-t_{u})}$. Since $\log y \leq \frac{y-1}{\sqrt{y}}$ for $y \geq 1$, taking $y = \frac{1}{t}$ and $y = \frac{1}{1-t}$ gives $h(t) \leq (1-t)\sqrt{t} + t\sqrt{1-t} \leq \sqrt{2 t(1-t)}$, so $\beta(K_{1}, K_{2}) \geq \sqrt{2} \big(\log 2 - I_{K}(U;X)\big)$.

Combining (i)--(iii), $\mathcal{C}(K) \leq -\log \big(2 (\log 2 - R_{0})^{2}\big) = 0.2506$ for every such $K$. Hence $C(R_{0}) \leq 0.2506 < 0.276$.

\subsection{Proof of Lemma~\ref{lem: K_R non empty compact convex}}\label{proof: K_R non empty compact convex}

We prove non-emptiness, compactness, and convexity in turn.

For non-emptiness. Let $K_0$ be a constant encoder, i.e.,
\[
K_0(1|x) = 1, \ K_0(u|x) = 0, \ \forall x \in \mathcal{X}, u = 2, \cdots, M.
\]
Then $I_{K_0}(U;X) = 0 \leq R$. Thus, $\mathcal{K}^{(M)}_{R}$ is non-empty.

For compactness. Let $\{K_{n} \}_{n=1}^{\infty} \subseteq \mathcal{K}^{(M)}_{R}$ be a sequence such that $K_{n} \to K^{*}$. Since $\Delta_{U|X}$ is the direct product of a finite number of simplexes, which are closed, indicating that $\Delta_{U|X}$ is closed, thus we have $K^{*} \in \Delta_{U|X}$. Moreover, $I(U;X)$ is a continuous functional of $K_{U|X}$. Then,
\[
I_{K^{*}}(U;X) = I_{\lim\limits_{n \to \infty} K_{n}}(U;X) = \lim\limits_{n \to \infty} I_{K_{n}}(U;X) \leq \lim\limits_{n \to \infty} R = R.
\]
Thus, $K^{*} \in \mathcal{K}^{(M)}_{R}$, i.e., $\mathcal{K}^{(M)}_{R}$ is closed. Since $\mathcal{K}^{(M)}_{R}$ is bounded, by Heine-Borel Theorem, $\mathcal{K}^{(M)}_{R}$ is compact.

For convexity, let $K_1, K_2 \in \mathcal{K}^{(M)}_{R}$ and 
$\theta \in [0,1]$. Define $K_{\theta}(u|x)=\theta K_1(u|x)+(1-\theta)K_2(u|x)$. For any $x \in \mathcal{X}$, we have $K_{\theta}(u|x)\geq 0$ and $\sum_{u\in\mathcal{U}_M}K_{\theta}(u|x) = \theta \sum_{u\in\mathcal{U}_M} K_1(u|x) + (1-\theta) \sum_{u\in\mathcal{U}_M}K_2(u|x) = 1$. Thus, $K_{\theta}\in\Delta_{U|X}$. For fixed $P_X$, the mutual information induced by encoder $K$ is
\[
I_{K}(U;X) = \mathrm{D}
\left(P_{XU}^{K} \middle\| P_X P_U^{K}\right),
\]
where
\[
P_{XU}^{K}(x,u)=P_X(x)K(u|x), \qquad P_U^{K}(u)=\sum_{x\in\mathcal{X}}P_X(x)K(u|x).
\]
Since both $P_{XU}^{K}$ and $P_XP_U^{K}$ are affine functions of $K$, we have
\[
P_{XU}^{K_{\theta}} = \theta P_{XU}^{K_1} + (1-\theta) P_{XU}^{K_2}, \quad P_X P_U^{K_{\theta}} = \theta P_XP_U^{K_1} + (1-\theta) P_X P_U^{K_2}.
\]
By the joint convexity of relative entropy,
\begin{align*}
I_{K_{\theta}}(U;X) & = \mathrm{D}
\left( P_{XU}^{K_{\theta}}
\middle\| P_X P_U^{K_{\theta}} \right) \\ & \leq
\theta \mathrm{D} \left( P_{XU}^{K_1} \middle\| P_X P_U^{K_1} \right) + (1-\theta) \mathrm{D} \left( P_{XU}^{K_2} \middle\| P_X P_U^{K_2} \right) \\ & =
\theta I_{K_1}(U;X) +  (1-\theta) I_{K_2}(U;X) \\ & \leq \theta R + (1-\theta) R = R,
\end{align*}
which gives $K_{\theta}\in\mathcal{K}^{(M)}_{R}$. Hence, $\mathcal{K}^{(M)}_{R}$ is convex.

\subsection{Proof of Lemma~\ref{lem: reduction to V}}\label{proof: reduction to V}

\textit{Proof of (i).} Fix $j \in \{1, \cdots, k\}$. On $\mathcal{X}_{j}$ we have $P_{X|1}(x) = \rho_{j} P_{X|0}(x)$, and summing over $\mathcal{X}_{j}$ gives $P_{V|1}(j) = \rho_{j} P_{V|0}(j)$. Hence $P_{X|0}(x) / P_{V|0}(j) = P_{X|1}(x) / P_{V|1}(j)$ for $x \in \mathcal{X}_{j}$, and averaging over $y$ with weights $P_{Y}(y)$ shows that this common value equals $P_{X}(x)/P_{V}(j)$, i.e.,
\begin{equation}
\label{eq: X indep Y given V}
    P_{X|y}(x) = P_{V|y}(j) \, \frac{P_{X}(x)}{P_{V}(j)}, \quad x \in \mathcal{X}_{j}, \ y \in \{0,1\}.
\end{equation}
Define the encoder $\bar{K}$ of $V$ with output alphabet $\mathcal{U}_{M}$ by
\[
\bar{K}(u|j) = \sum_{x \in \mathcal{X}_{j}} \frac{P_{X}(x)}{P_{V}(j)} K(u|x),
\]
which is well defined since $\sum_{u} \bar{K}(u|j) = \sum_{x \in \mathcal{X}_{j}} P_{X}(x) / P_{V}(j) = 1$. By \eqref{eq: X indep Y given V},
\[
P^{\bar{K}}_{U|y}(u) = \sum_{j=1}^{k} \sum_{x \in \mathcal{X}_{j}} P_{V|y}(j) \frac{P_{X}(x)}{P_{V}(j)} K(u|x) = \sum_{x \in \mathcal{X}} P_{X|y}(x) K(u|x) = P^{K}_{U|y}(u).
\]
Moreover, the joint distribution of $(V,U)$ under $K$ is $\sum_{x \in \mathcal{X}_{j}} P_{X}(x) K(u|x) = P_{V}(j) \bar{K}(u|j)$, which is the one induced by $\bar{K}$ with input $P_{V}$. Since $V$ is a function of $X$,
\[
I_{\bar{K}}(U;V) = I_{K}(U;V) \leq I_{K}(U;X,V) = I_{K}(U;X).
\]

\textit{Proof of (ii).} Grouping the symbols of $\mathcal{X}$ by class,
\[
P^{K}_{U|y}(u) = \sum_{j=1}^{k} \sum_{x \in \mathcal{X}_{j}} P_{X|y}(x) \bar{K}(u|j) = \sum_{j=1}^{k} P_{V|y}(j) \bar{K}(u|j) = P^{\bar{K}}_{U|y}(u).
\]
Under $K$, the output $U$ depends on $X$ only through $V = \nu(X)$, so $I_{K}(U;X|V) = 0$, and the joint distribution of $(V,U)$ is $P_{V}(j) \bar{K}(u|j)$. Hence
\[
I_{K}(U;X) = I_{K}(U;X,V) = I_{K}(U;V) + I_{K}(U;X|V) = I_{\bar{K}}(U;V).
\]

\subsection{Proof of Lemma~\ref{lem: uniform bound of w(u)}}\label{proof: uniform bound of w(u)}

The Markovity gives
\[
P^{K}_{U|y}(u) = \sum_{x \in \mathcal{X}} K(u|x) P_{X|y}(x), \quad y \in \{0,1\}.
\]
Then,
\[
\frac{P^{K}_{U|1}(u)}{P^{K}_{U|0}(u)} = \frac{\sum_{x \in \mathcal{X}} K(u|x) P_{X|1}(x)}{\sum_{x \in \mathcal{X}} K(u|x) P_{X|0}(x)}.
\]
Since $K_{U|X} \in \mathcal{K}^{(M)}_{R,+}$, and by \eqref{eq: condition no instantly recognizable x}, applying the weighted average inequality (Fact \ref{fact: Weighted average inequality}) derives
\[
\ell = \min_{x \in \mathcal{X}} \frac{P_{X|1}(x)}{P_{X|0}(x)} \leq \frac{P^{K}_{U|1}(u)}{P^{K}_{U|0}(u)} \leq \max_{x \in \mathcal{X}} \frac{P_{X|1}(x)}{P_{X|0}(x)} = L.
\]

\subsection{Proof of Lemma~\ref{lem: algorithm GBA uniform bounded}}\label{proof: algorithm GBA uniform bounded}

Applying Young's inequality (Fact \ref{fact: Young's inequality}) gives
    \begin{align*}
        C^{s}_{u,x} & = s w(u) P_{X|0}(x) + (1-s) w(u)^{-\frac{s}{1-s}} P_{X|1}(x) \\ & \geq \big(w(u) P_{X|0}(x) \big)^{s} \big( w(u)^{-\frac{s}{1-s}} P_{X|1}(x) \big)^{1-s} \\ & = P_{X|0}(x)^{s}P_{X|1}(x)^{1-s},
    \end{align*}
    with the lower bound strictly positive and independent of the encoder.

    By Lemma \ref{lem: uniform bound of w(u)}, and $\ell \leq 1 \leq L$ (since $\sum_{x} \rho(x) P_{X|0}(x) = \sum_{x} P_{X|1}(x) = 1$), we have
    \[
    w(u) \leq L^{1-s} \leq L \leq \Gamma, \quad w(u)^{-\frac{s}{1-s}} \leq \ell^{-s} \leq \ell^{-1} \leq \Gamma.
    \]
    Substituting into \eqref{eq: parameter matrix Cs} derives
    \[
     C^{s}_{u,x} \leq \Gamma \big(sP_{X|0}(x) +(1-s) P_{X|1}(x) \big).
    \]

\subsection{Proof of Theorem~\ref{the: hatL global converge}}\label{proof: hatL global converge}

Write $q_{t} = P^{(t)}_{U}$, $K_{t} = K^{(t)}_{U|X}$, and $Z_{t}(x) = \sum_{u'} q_{t}(u') \exp\big(-C^{s}_{u',x}/(\lambda P_{X}(x))\big)$. Since $q_{0}$ has full support and every exponential factor is positive, induction on $t$ gives $q_{t}(u) > 0$ for all $t \geq 0$ and $K_{t}(u|x) > 0$ for all $t \geq 1$ and all $(u,x)$. Hence every relative entropy below is finite. Note also that $q_{t} = P^{K_{t}}_{U}$ for every $t \geq 0$.

\textit{One-step inequality.} Let $K \in \mathcal{K}^{(M)}_{\infty}$ and $q = P^{K}_{U}$. The update in Algorithm \ref{alg: GBA for K_U|X} gives $C^{s}_{u,x} = -\lambda P_{X}(x) \log \big(Z_{t}(x) K_{t+1}(u|x) / q_{t}(u)\big)$. Substituting it into \eqref{eq: definition hatL},
\[
\hat{\mathcal{L}}_{\lambda}(K,q_{t}) = \lambda \sum_{x} P_{X}(x) \Big( \mathrm{D}\big(K(\cdot|x) \,\|\, K_{t+1}(\cdot|x)\big) - \log Z_{t}(x) \Big).
\]
Taking $K = K_{t+1}$ and subtracting, then applying Fact \ref{fact: joint convexity of relative entropy} with $q = \sum_{x} P_{X}(x) K(\cdot|x)$ and $q_{t+1} = \sum_{x} P_{X}(x) K_{t+1}(\cdot|x)$,
\[
\hat{\mathcal{L}}_{\lambda}(K,q_{t}) - \hat{\mathcal{L}}_{\lambda}(K_{t+1},q_{t}) = \lambda \sum_{x} P_{X}(x) \mathrm{D}\big(K(\cdot|x) \,\|\, K_{t+1}(\cdot|x)\big) \geq \lambda \mathrm{D}(q \| q_{t+1}).
\]
By \eqref{eq: hatL = L + D}, $\mathcal{L}_{\lambda}(K_{t+1}) \leq \hat{\mathcal{L}}_{\lambda}(K_{t+1},q_{t})$ and $\hat{\mathcal{L}}_{\lambda}(K,q_{t}) = \mathcal{L}_{\lambda}(K) + \lambda \mathrm{D}(q \| q_{t})$. Hence
\begin{equation}
\label{eq: GBA one-step}
    \mathcal{L}_{\lambda}(K_{t+1}) \leq \mathcal{L}_{\lambda}(K) + \lambda \big( \mathrm{D}(q \| q_{t}) - \mathrm{D}(q \| q_{t+1}) \big), \quad \forall K \in \mathcal{K}^{(M)}_{\infty}.
\end{equation}

\textit{Monotonicity and rate.} Taking $K = K_{t}$, for which $q = q_{t}$, \eqref{eq: GBA one-step} gives $\mathcal{L}_{\lambda}(K_{t+1}) \leq \mathcal{L}_{\lambda}(K_{t})$. Taking $K = K^{*}$ with $q^{*} = P^{K^{*}}_{U}$, summing \eqref{eq: GBA one-step} over $t = 0, \cdots, T-1$, and using the monotonicity,
\[
T \big( \mathcal{L}_{\lambda}(K_{T}) - \mathcal{L}^{*}_{\lambda} \big) \leq \sum_{t=1}^{T} \big( \mathcal{L}_{\lambda}(K_{t}) - \mathcal{L}^{*}_{\lambda} \big) \leq \lambda \big( \mathrm{D}(q^{*} \| q_{0}) - \mathrm{D}(q^{*} \| q_{T}) \big) \leq \lambda \mathrm{D}(q^{*} \| q_{0}),
\]
which proves \eqref{eq: GBA rate}.

\textit{Convergence of the iterates.} Since $\mathcal{K}^{(M)}_{\infty}$ is compact, there is a subsequence $K_{t_{j}} \to \bar{K}$, and then $q_{t_{j}} \to \bar{q} = P^{\bar{K}}_{U}$. By \eqref{eq: GBA rate} and the continuity of $\mathcal{L}_{\lambda}$, $\bar{K}$ is a global minimizer of $\mathcal{L}_{\lambda}$. Applying \eqref{eq: GBA one-step} with $K = \bar{K}$ and using $\mathcal{L}_{\lambda}(K_{t+1}) \geq \mathcal{L}_{\lambda}(\bar{K})$ shows that $\mathrm{D}(\bar{q} \| q_{t})$ is nonincreasing in $t$. Moreover, $\mathrm{D}(\bar{q} \| q_{t_{j}}) = \sum_{u: \bar{q}(u) > 0} \bar{q}(u) \log \frac{\bar{q}(u)}{q_{t_{j}}(u)} \to 0$, since $q_{t_{j}}(u) \to \bar{q}(u) > 0$ in every nonzero term. Hence the whole sequence $\mathrm{D}(\bar{q} \| q_{t}) \to 0$, and $q_{t} \to \bar{q}$ by Fact \ref{fact: Pinsker's inequality}. Since $Z_{t}(x) \geq \min_{u} \exp\big(-C^{s}_{u,x}/(\lambda P_{X}(x))\big) > 0$, $K_{t+1}$ is a continuous function of $q_{t}$ through the update in Algorithm \ref{alg: GBA for K_U|X}. Therefore $K_{t}$ converges, and its limit is $\bar{K}$. Finally, by \eqref{eq: hatL = L + D}, $\hat{\mathcal{L}}_{\lambda}(\bar{K}, \bar{q}) = \mathcal{L}_{\lambda}(\bar{K}) = \mathcal{L}^{*}_{\lambda}$, so $(K_{t}, q_{t}) \to (\bar{K}, P^{\bar{K}}_{U})$, a global minimizer of $\hat{\mathcal{L}}_{\lambda}$.

\subsection{Proof of Lemma~\ref{lem: bisection search of lambda}}\label{proof: bisection search of lambda}

\textit{Proof of \ref{item: bisection monotone}.} Let $0 \leq \lambda_1 < \lambda_2$, and write $K_i = K(\lambda_i)$, $G_i = G(K_i)$ and $I_i = I(\lambda_i)$ for $i = 1, 2$. Since $K_1$ minimizes $\mathcal{L}_{\lambda_1}$ and $K_2$ minimizes $\mathcal{L}_{\lambda_2}$ over $\mathcal{K}^{(M)}_{\infty}$,
\[
G_1 + \lambda_1 I_1 \leq G_2 + \lambda_1 I_2, \qquad G_2 + \lambda_2 I_2 \leq G_1 + \lambda_2 I_1.
\]
Adding the two inequalities gives $(\lambda_2 - \lambda_1)(I_1 - I_2) \geq 0$, hence $I_1 \geq I_2$.

\textit{Proof of \ref{item: bisection certificate}.} By Lemma \ref{lem: K_R non empty compact convex}, $\mathcal{K}^{(M)}_{R}$ is non-empty and compact, and $G$ is continuous, so $G^{*}$ is attained by some $K^{\circ} \in \mathcal{K}^{(M)}_{R}$. Since $K(\lambda)$ minimizes $\mathcal{L}_{\lambda}$ over $\mathcal{K}^{(M)}_{\infty} \supseteq \mathcal{K}^{(M)}_{R}$,
\[
G(K(\lambda)) + \lambda I(\lambda) \leq G(K^{\circ}) + \lambda I_{K^{\circ}}(U;X) \leq G^{*} + \lambda R,
\]
which gives the upper bound in \eqref{eq: bisection certificate}. If $I(\lambda) \leq R$, then $K(\lambda) \in \mathcal{K}^{(M)}_{R}$, so that $G(K(\lambda)) \geq G^{*}$, which gives the lower bound. If $I(\lambda) = R$ or $\lambda = 0$, the right-hand side of \eqref{eq: bisection certificate} vanishes, and $G(K(\lambda)) = G^{*}$.

\textit{Proof of \ref{item: bisection threshold}.} Let $\beta = \sum_{x} P_{X|0}(x)^{s} P_{X|1}(x)^{1-s}$, which is strictly positive by \eqref{eq: condition no instantly recognizable x}. By the lower bound in Lemma \ref{lem: algorithm GBA uniform bounded}, for any $K \in \mathcal{K}^{(M)}_{\infty}$,
\[
G(K) \geq \sum_{x} \Big( \sum_{u} K(u|x) \Big) P_{X|0}(x)^{s} P_{X|1}(x)^{1-s} = \beta.
\]
Let $K_{0}$ be the constant encoder in Appendix \ref{proof: K_R non empty compact convex}, which satisfies $I_{K_0}(U;X) = 0$. By the upper bound in Lemma \ref{lem: algorithm GBA uniform bounded},
\[
G(K_0) = \sum_{x} C^{s}_{1,x} \leq \Gamma \sum_{x} \big( s P_{X|0}(x) + (1-s) P_{X|1}(x) \big) = \Gamma.
\]
Hence, for any $\lambda > 0$, the minimality of $K(\lambda)$ gives
\[
\beta + \lambda I(\lambda) \leq \mathcal{L}_{\lambda}(K(\lambda)) \leq \mathcal{L}_{\lambda}(K_0) = G(K_0) \leq \Gamma,
\]
that is, $I(\lambda) \leq (\Gamma - \beta) / \lambda$, which is at most $R$ whenever $\lambda \geq (\Gamma - \beta)/R$. Now define
\[
\lambda^{*} \stackrel{\Delta}{=} \inf \big\{ \lambda \geq 0 : I(\lambda) \leq R \text{ for some choice of } K(\lambda) \big\},
\]
so that $\lambda^{*} \leq (\Gamma - \beta)/R < \Gamma/R$. If $\lambda > \lambda^{*}$, there exist $\lambda_{1} < \lambda$ and a choice of $K(\lambda_1)$ with $I(\lambda_1) \leq R$, and \ref{item: bisection monotone} gives $I(\lambda) \leq I(\lambda_1) \leq R$ for every choice of $K(\lambda)$. If $\lambda < \lambda^{*}$, then by the definition of $\lambda^{*}$, $I(\lambda) > R$ for every choice of $K(\lambda)$.

\textit{Consequences for the bisection.} The bisection sets $\lambda_{\text{small}} \leftarrow \lambda$ only when $I(\lambda) > R$, which by \ref{item: bisection threshold} forces $\lambda \leq \lambda^{*}$, and sets $\lambda_{\text{big}} \leftarrow \lambda$ only when $I(\lambda) \leq R$, which forces $\lambda \geq \lambda^{*}$. Initially, $\lambda_{\text{small}} = 0 \leq \lambda^{*} < \Gamma/R = \lambda_{\text{big}}$, and $I(\Gamma/R) \leq R$ by \ref{item: bisection threshold}. Hence $\lambda^{*} \in [\lambda_{\text{small}}, \lambda_{\text{big}}]$ and $I(\lambda_{\text{big}}) \leq R$ at every step, and since each step halves the bracket, $\lambda_{\text{big}} - \lambda_{\text{small}} = \Gamma / (R \, 2^{n})$ after $n$ steps.

\end{document}